\documentclass[12pt]{article}

\pdfoutput=1

\usepackage{latexsym,amsfonts,bm,epsfig,amsmath,natbib,authblk,amsthm,thmtools,amssymb}

\usepackage{float}
\usepackage{booktabs}
\usepackage{graphicx}
\usepackage[margin=1cm]{caption}

\floatstyle{plain}
\newfloat{Algorithm}{thp}{lop}
\floatname{Algorithm}{Algorithm}
\newenvironment{wideenumerate}{\enumerate\addtolength{\itemsep}{5pt}}{\endenumerate}

\usepackage[plain,noend]{algorithm2e}

\makeatletter
\renewcommand{\algocf@captiontext}[2]{#1\algocf@typo. \AlCapFnt{}#2} 
\def\@algocf@capt@plain{top}
\renewcommand{\algocf@makecaption}[2]{%
  \addtolength{\hsize}{\algomargin}%
  \sbox\@tempboxa{\algocf@captiontext{#1}{#2}}%
  \ifdim\wd\@tempboxa >\hsize
   \hskip .5\algomargin%
    \parbox[t]{\hsize}{\algocf@captiontext{#1}{#2}}
 \else%
   \global\@minipagefalse%
   \hbox to\hsize{\box\@tempboxa}
 \fi%
  \addtolength{\hsize}{-\algomargin}%
}
\makeatother

\newcommand{\ignore}[1]{}

\newcommand{\dt}{\text{d}}

\newtheorem{assumption}{Assumption}
\newtheorem{remark}{Remark}
\newtheorem{theorem}{Theorem}
\newtheorem{corollary}{Corollary}
\newtheorem{lemma}{Lemma}

\begin{document}

\title{Bayesian inference using synthetic likelihood:  asymptotics and adjustments}
\date{\empty}

\author[1,6]{David T. Frazier}
\author[2,3]{David J. Nott\thanks{Corresponding author:  standj@nus.edu.sg}}
\author[4,6]{Christopher Drovandi}
\author[5,6]{Robert Kohn}
\affil[1]{Department of Econometrics and Business Statistics, Monash University, Clayton VIC 3800, Australia}
\affil[2]{Department of Statistics and Applied Probability, National University of Singapore, Singapore 117546}
\affil[3]{Operations Research and Analytics Cluster, National University of Singapore, Singapore 119077}
\affil[4]{School of Mathematical Sciences, Queensland University of Technology, Brisbane 4000 Australia}
\affil[5]{Australian School of Business, School of Economics, University of New South Wales, Sydney NSW 2052, Australia}
\affil[6]{Australian Centre of Excellence for Mathematical and Statistical Frontiers (ACEMS)}

\maketitle

\vspace{-0.3in}

\begin{abstract}
Implementing Bayesian inference is often computationally challenging in applications involving complex models, and
sometimes calculating the likelihood itself is difficult.  Synthetic likelihood is one
approach for carrying out inference when the likelihood
is intractable, but it is straightforward to simulate from the model.  The method
constructs an approximate likelihood by taking a vector
summary statistic as being multivariate normal, with the unknown mean and covariance matrix estimated by simulation
for any given parameter value. Previous empirical research demonstrates that the Bayesian implementation of synthetic likelihood can be more computationally efficient than approximate Bayesian computation, a popular likelihood-free method, in the presence of a high-dimensional summary statistic.  Our article makes three contributions.  The first shows that if the summary statistic satisfies
a central limit theorem, then the synthetic likelihood posterior is asymptotically normal and yields credible sets with the correct level of frequentist coverage.  This result is similar to that obtained by approximate Bayesian computation.  The second contribution
compares the computational efficiency of Bayesian synthetic likelihood and approximate Bayesian computation using
the acceptance probability for rejection and importance sampling algorithms with a ``good'' proposal distribution.
We show that Bayesian synthetic likelihood is computationally more efficient than approximate Bayesian computation, and behaves similarly to
regression-adjusted approximate Bayesian computation.  Based on the
asymptotic results, the third contribution proposes using adjusted inference methods when a possibly misspecified form is assumed
for the covariance matrix of the synthetic likelihood, such as diagonal or a factor model, to speed up the computation.  The methodology is illustrated with some simulated and real examples.

\end{abstract}

\section{Introduction}\label{sec:Intro}

Synthetic likelihood is a popular method used in likelihood-free inference when the likelihood is intractable, but it is possible to simulate
from the model for any given parameter value.  The method takes a vector summary statistic that is assumed to be informative about the parameter
and assumes it is multivariate normal, estimating the unknown mean and covariance matrix by simulation to produce an approximate likelihood function.  \citet{price+dln16} provide empirical and preliminary theoretical evidence that Bayesian synthetic likelihood (BSL) can perform favourably compared to approximate Bayesian computation \citep[ABC,][]{sisson2018}, a more mature likelihood-free method that has been subjected to extensive theoretical examination.  The performance gains of BSL are particularly noticeable in the presence of a regular, high-dimensional summary statistic. Given the promising empirical performance of BSL, it is important to study its theoretical properties.

This article makes three contributions.  First, it
investigates the asymptotic properties of synthetic likelihood when the summary statistic satisfies a central limit theorem.  The conditions required for the results are similar to those in \citet{frazier+mrr18} in the asymptotic analysis of ABC algorithms, but with an additional assumption controlling the uniform behaviour of
summary statistic covariance matrices.  Under appropriate conditions,
the posterior density is asymptotically normal and it quantifies uncertainty accurately,
similarly to ABC approaches \citep{li+f18b,li+f18a,frazier+mrr18}.

The second contribution is to show that a rejection sampling
BSL algorithm has a non-negligible acceptance probability for a ``good'' proposal
density.  A similar ABC algorithm
has an acceptance probability that goes to zero asymptotically,
and synthetic likelihood performs similarly to regression-adjusted ABC \citep{li+f18b,li+f18a}.

The third contribution considers
situations where a parsimonious but misspecified form is assumed
for the covariance matrix of the summary statistic, such as a diagonal matrix or a factor model,
to speed up the computation.
For example, \citet{priddle+sfd19} show that for a diagonal covariance matrix, the number of simulations need only grow linearly with the summary statistic dimension
to control the variance of the synthetic likelihood estimator, as opposed to quadratically for the full covariance matrix.
This is especially important for models
where simulation of summary statistics is expensive.  We use our asymptotic results to motivate
sandwich-type variance adjustments to account for the misspecification and implement these in some examples.
The adjustments just discussed are also potentially useful when the model for the original data
is misspecified and  we wish to carry out inference
for the pseudo-true parameter value with the
data generating density closest to the truth; Section 2.1  elaborates on these ideas.

For the adjustment methods to be valid, it is important that
the summary statistic satisfies a central limit theorem, so that we can make use of the asymptotic normality of the posterior density.  This means that these
adjustments are not useful for correcting for the effects of violating
the normality assumption for the summary statistic.  \cite{muller13} considers some related methods,
although not in the context of synthetic likelihood or likelihood-free inference.  \citet{frazier+rr17}
studies the
consequences of misspecification for ABC approaches to likelihood-free
inference.

\citet{wood10} introduced the synthetic likelihood and used it for approximate
(non-Bayesian) inference.  \citet{price+dln16} discussed
Bayesian implementations focusing on efficient computational methods.  They also show that the synthetic likelihood
scales more easily to high-dimensional problems and that it is easier to tune than competing approaches such as ABC.

There is much recent development of innovative methodology for
accelerating computations for synthetic likelihood and related methods \citep{meeds+w14,wilkinson14,gutmann+c15,everitt17,ong+ntsd16,ong+ntsd18,An2016,priddle+sfd19}.  However,
there is also interest in weakening the normality assumption on which
the synthetic likelihood is based.  This  led several authors to use other surrogate likelihoods for more flexible summaries.  
For example, \citet{Fasiolo2016} consider extended saddlepoint approximations,
\citet{Dutta2016} consider a logistic regression approach for likelihood estimation, and \cite{an+nd18} consider semiparametric
density estimation with flexible marginals and a Gaussian copula dependence structure.
\citet{mengersen+pr13} and \citet{chaudhuri+gnp18} consider empirical likelihood approaches.
An encompassing framework for many of these suggestions is the parametric
Bayesian indirect likelihood of \citet{drovandi+pl15}.

As mentioned above,
the adjustments for misspecification developed here do not
contribute to this literature on robustifying synthetic likelihood inferences to non-normality
of the summary statistics, as they can only be justified when a central limit
theorem holds for the summary statistic.  Bayesian analyses involving pseudo-likelihoods have been considered
in the framework of Laplace-type estimators discussed in \cite{chernozhukov+h03},
but their work does not deal with settings where
the likelihood itself must be estimated using Monte Carlo.  \cite{forneron+n18}
developed some theory connecting
ABC approaches with simulated minimum distance methods widely used in econometrics, and their discussion is also relevant to simulation versions of Laplace-type estimators.

\section{Bayesian synthetic likelihood}\label{asymptotic}
Let $y=(y_1,\dots,y_n)^{\intercal}$ denote the observed data and define $P^{(n)}_0$ as the true distribution generating $y$. The model $P^{(n)}_0$ is approximated using a parametric family of models $\{P^{(n)}_\theta:\theta\in\Theta\subset\mathbb{R}^{d_\theta}\}$, and $\Pi$ denotes the prior distribution over $\Theta$, with density $\pi(\theta)$.
 We are interested in situations where, due to the complicated nature of the model, the likelihood of $P^{(n)}_\theta$ is intractable. In such cases, approximate methods such as BSL can be used to conduct inference on the unknown $\theta$.

Like the ABC method, 
BSL is most commonly implemented by replacing the observed data $y$ by a low-dimensional vector of 
summary statistics. Throughout, we let the function $S_n:\mathbb{R}^n\rightarrow\mathbb{R}^d$, 
$d\ge d_\theta$, represents the chosen vector (function) of summary statistics. For a given model 
$P^{(n)}_\theta$, let $z=(z_1,\dots,z_n)^{\intercal}$ denote data generated 
under the model $P^{(n)}_\theta$, and let $b_{}(\theta):=\mathbb{E}\{S_n(z)|\theta\}$ and
$\Sigma_n(\theta):=\text{var}\{S_n(z)|\theta\}$ denote the mean and variance of the summaries calculated under $P^{(n)}_\theta$; the map $\theta\mapsto b(\theta)$ may technically depend on $n$. However, if the data are independent and identically distributed or weakly dependent, and if $S_n$ can be written as an average, $b(\theta)$ will not meaningfully depend on $n$.  As the vast majority of summaries used in BSL satisfy this scenario, neglecting the potential dependence on $n$ is reasonable.

The synthetic likelihood method approximates the intractable likelihood of $S_n(z)$ by a normal likelihood. If $b_{}(\theta)$ and $\Sigma_n(\theta)$ are known, then the synthetic likelihood
is
\begin{align*}
g_n(S_n|\theta) & := N\left\{S_n;b(\theta),\Sigma_n(\theta)\right\};
\end{align*}
here, and below,   $N(\mu,\Sigma)$ denotes a  normal distribution with mean $\mu$ and
covariance matrix $\Sigma$, and $N(x;\mu,\Sigma)$ is its density function evaluated at $x$.

The idealized BSL posterior using known $b(\theta)$ and $\Sigma_n(\theta)$ is
$$\pi(\theta|S_n)=\frac{g_n(S_n|\theta)\pi(\theta)}{\int_\Theta g_n(S_n|\theta)\pi(\theta)\dt \theta};$$
Markov chain Monte Carlo (MCMC) is used to obtain draws from the target posterior $\pi(\theta|S_n)$, which we assume exists for all $n$. However, outside of toy examples, posterior inference based on $\pi(\theta|S_n)$ is infeasible since $b(\theta)$ and $\Sigma_n(\theta)$ can only be analytically calculated if the mean and variance of $S_n(z)$ is known.

Therefore, BSL is generally implemented by replacing $b(\theta)$ and $\Sigma_n(\theta)$ with estimates $\widehat{b}_n(\theta)$ and
$\widehat{\Sigma}_n(\theta)$.  To obtain these estimates, we generate $m$ independent summary statistics
$\{S(z^i)\}_{i=1}^{m}$, where $z^i\sim P^{(n)}_\theta$, and take $\widehat{b}_n(\theta)$ as the sample mean of the $S_n(z^i)$ and
$\widehat{\Sigma}_n(\theta)$ as their sample covariance matrix.
The notation does not show the dependence of $\widehat{b}_n(\theta)$ and $\widehat{\Sigma}_n(\theta)$ on $m$, since $m$ is later taken as a function of $n$. In practical applications of BSL, the use of variance estimates other than $\widehat{\Sigma}_n(\theta)$ is common (e.g.\ \citealp{An2016}, \citealp{ong+ntsd18} and \citealp{priddle+sfd19}). To encapsulate these and other situations, we take ${\Delta}_n(\theta)$ to be a general covariance matrix estimator.

When $b(\theta)$ and $\Sigma_n(\theta)$ are replaced with estimates, BSL attempts to sample the following posterior target
\begin{align}\label{eq:perturbedpost}
\widehat{\pi}(\theta|S_n)& \propto \pi(\theta) \widehat{g}_n(S_n|\theta),
\end{align}where, for $q_n(\cdot|\theta)$ the density of the simulated summary statistics under $P^{(n)}_\theta$,
\begin{align}
\widehat{g}_n(S_n|\theta) & := \int N\{S_n;\widehat{b}_n(\theta),{\Delta}_n(\theta)\} \prod_{i=1}^m q_n\{S(z^i)|\theta\}\,\dt S(z^1)\,\dots\, \dt S(z^m) \label{noisySL}.
\end{align}
Noting that an unbiased estimator of  $\widehat{g}_n(S_n|\theta)$ can be obtained by taking a single draw of $S_n(z^i)\sim q_n(\cdot|\theta)$, and following arguments in Andrieu and Roberts (2009), a pseudo-marginal algorithm employing an estimator of $\widehat{g}_n(S_n|\theta)$ results in sampling from the posterior density $\widehat{\pi}(\theta|S_n)$ in \eqref{eq:perturbedpost} under reasonable integrability assumptions. Therefore, estimation of $b(\theta)$ and $\Sigma_n(\theta)$ ensures that the BSL posterior target, $\widehat{\pi}(\theta|S_n)$, and the idealized BSL posterior,
$\pi(\theta|S_n)$, will differ. 

Under idealized, but useful assumptions, \cite{pitt+sgk12}, \cite{doucet+pdk15} and \cite{sherlock2015}
choose the number of samples $m$ in pseudo-marginal MCMC to optimize the time normalized variance of
the posterior mean estimators.
They show that a good choice of $m$ occurs (for a given $\theta$)
when the variance $\sigma^2(\theta)$  of the log of the likelihood estimator  lies between 1 and 3, with a value of 1 suitable for a very good  proposal, i.e.,  close to the posterior,  and around 3 for an inefficient proposal, e.g. a random walk. \cite{deligiannidis2018correlated}
propose a correlated pseudo-marginal sampler that tolerates
a much greater value of $\sigma^2(\theta)$, and hence a much smaller value of $m$,
when the random numbers used to construct the estimates of the likelihood at both the current and
proposed values of $\theta$ are correlated; see also \cite{tran2016block}
for an alternative construction of a correlated block pseudo-marginal sampler.

Here, the perturbed BSL target
is \eqref{noisySL} and the log of its estimate is,
\begin{equation}\label{eq:gest}
-\frac{1}{2}\log\left\{\left|\Delta_n(\theta)\right|\right\}-\frac{1}{2}\left\{S_n-\widehat{b}_n(\theta)\right\}^{\intercal}{\Delta}^{-1}_n(\theta)\left\{S_n-\widehat{b}_n(\theta)\right\},
\end{equation}
omitting additive terms  not depending on $\theta$. 
It is straightforward to incorporate either
the correlated or block pseudo-marginal approaches into the estimation and show that \eqref{eq:gest} is bounded in a neighbourhood of $\theta_0$ if the eigenvalues of $\Delta_n(\theta)$ are bounded away from zero, suggesting that the variance of the 
log of the estimate of the synthetic likelihood \eqref{eq:gest}
 will not have a high variance  in practice. We do not 
  not derive theory for how to select $m$ optimally because that requires taking account of the bias and variance of the synthetic likelihood, which is unavailable in general due to the intractability of the likelihood. However, our empirical
 work limits $\sigma^2(\theta)$ to lie between 1 and 3, which produces good results. \cite{price+dln16} find in their examples that the approximate posterior in \eqref{eq:perturbedpost} depends only weakly on
the choice of $m$, and hence they often  choose a small value of $m$  for faster computation.

The BSL posterior in \eqref{eq:perturbedpost} is constructed from three separate approximations: (1) the representation of the observed data $y$ by the summaries $S_n(y)$; (2) the approximation of the unknown distribution for the summaries by a Gaussian with unknown mean $b(\theta)$ and covariance $\Sigma_n(\theta)$; (3) the approximation of the unknown mean and covariance by the  estimates $\widehat{b}_n(\theta)$ and ${\Delta}_n(\theta)$. 

Given the various approximations involved in BSL, it is critical to understand precisely how these approximations impact the resulting inferences on the unknown parameters $\theta$. {In practice,  understanding how $m$ and ${\Delta}_n(\theta)$ affect the resulting inferences is particularly important. The larger $m$, the more time consuming is the computation of the BSL posterior. Replacing $\Sigma_n(\theta)$, the covariance of the summaries, by $\Delta_n(\theta)$ means that the posterior may not reliably quantify uncertainty (if $\Delta_n(\theta)$ is not carefully chosen).}
Any theoretical analysis of the BSL posterior is made difficult by the intractability  of $P^{(n)}_\theta$, and ensures that exploring the finite-sample behavior of the BSL likelihood estimate in \eqref{eq:gest}, and ultimately $\widehat{\pi}(\theta|S_n)$, is difficult in general problems. We therefore use asymptotic methods
to study the impact of the various approximation within BSL on the resulting inference for $\theta$.

\section{Asymptotic Behavior of BSL}
This section contains several results that disentangle the impact of the previously mentioned approximations used in BSL. These demonstrate that, under regularity conditions, BSL delivers inferences that are just as reliable as other approximate Bayesian methods, such as ABC. Moreover, unlike the commonly applied accept/reject ABC, the acceptance probability obtained by running BSL does not converge to zero as the sample size increases, and is not affected by the number of summaries {(assuming they are of fixed dimension, i.e., $d=\text{dim}(S_n)$ does not change as $n$ increases)}.

A Bernstein von-Mises result is first proved 
and is then used  to deduce asymptotic normality of the BSL posterior mean. Using these results, we can demonstrate that valid uncertainty quantification in BSL requires: (1) $m\rightarrow\infty$ as $n\rightarrow\infty$; (2) the chosen covariance matrix used in BSL, ${\Delta}_n(\theta)$, must be a consistent estimator for the asymptotic variance of the observed summaries $S_n(y)$.

Some notation is now defined to make the results below easier to state and follow.
For $x\in\mathbb{R}^{d}$, $\| x\| $ denotes the Euclidean norm of $x$. For any matrix $M\in\mathbb{R}^{d\times d}$, we define $|M|$ as the determinant of $M$, 
and, with some abuse of notation, 
 let $\|M\|$ denote any convenient matrix norm of $M$; the choice of $\|\cdot\|$ is immaterial since we will always be working with matrices of fixed $d\times d$ dimension, so that all matrix norms are equivalent. Let $\text{Int}(\Theta)$ denote the interior of the set $\Theta$. 
 Throughout, let  $C$
denote a generic positive constant that can change with each use. 
For real-valued sequences $\{a_{n}\}_{n\geq 1}$ and
$\{b_{n}\}_{n\geq 1}$: $a_{n}\lesssim b_{n}$ denotes $a_{n}\leq Cb_{n}$ for
some finite $C>0$ and all $n$ large, $a_{n}\asymp b_{n}$ implies $a_{n}\lesssim b_{n}$ and $b_n \lesssim a_{n}$. For $x_{n}$ a random variable, $x_{n}=o_{p}(a_{n})$ if
$\lim_{n\rightarrow \infty }\text{pr} (|x_{n}/a_{n}|\geq C)=0$ for any $C>0, $
and $x_{n}=O_{p}(a_{n})$ if for any $C>0$ there exists a finite $M>0$ and a
finite $n'$ such that, for all $n>n'$, $\text{pr}(|x_{n}/a_{n}|\geq M)\leq C$. All limits are taken as $n\rightarrow\infty$, so that, when no confusion will result, we use $\lim_{n}$ to denote $\lim_{n\rightarrow\infty}$. The notation $\Rightarrow$ denotes weak convergence.  The Appendix
contains all the proofs.

\subsection{Asymptotic Behavior of the BSL Posterior}
{This section establishes the asymptotic behavior of the BSL posterior $\widehat{\pi}(\theta|S_n)$ in equation \eqref{eq:perturbedpost}.} We do not assume that
${\Delta}_n(\theta)$ is a consistent estimator of $\Sigma_n(\theta)$
to allow the synthetic likelihood covariance to be ``misspecified''.
The following regularity conditions are assumed on $S_n$, $b(\theta)$ and ${\Delta}_n(\theta)$.

\begin{assumption}\label{ass:one}
There exists a sequence of positive real numbers $v_n$ diverging to $\infty$ and a vector $b_0\in\mathbb{R}^d$, $d\ge d_\theta$, such that $V_0:=\lim_{n}\text{var}\left\{v_n(S_n-b_0)\right\}$ exists and $$v_n\left(S_ { n }-b_0\right)\Rightarrow N\left(0, V_0\right),\text{ under } P^{(n)}_0.$$
\end{assumption}

\begin{assumption}\label{ass:four}
	(i) The map $\theta\mapsto{b}(\theta)$
		is continuous, and there exists a unique $\theta_0\in\text{Int}(\Theta)$, such that $b(\theta_0)=b_0$; (ii) for some $\delta>0$, and all $\|\theta-\theta_0\|\le\delta$, the Jacobian $\nabla{b}(\theta)$ exists and is continuous, and $\nabla{b}(\theta_0)$ has full column rank $d_\theta$.
\end{assumption}

\begin{assumption}\label{ass:two}The following conditions
are satisfied for some $\delta>0$: (i) for $n$ large enough, the matrix $v_n^2{\Delta}_n(\theta)$ is positive-definite for all $\|\theta-\theta_0\|\le\delta$; (ii) there exists some matrix $\Delta(\theta)$, positive semi-definite uniformly over $\Theta$, and such that $\sup_{\theta\in\Theta}\|v^2_n{\Delta}^{}_n(\theta)-\Delta^{}(\theta)\|=o_{p}(1)$, and, for all $\|\theta-\theta_0\|\le\delta$, $\Delta(\theta)$ is continuous and positive-definite; (iii) for any $\epsilon>0$, $\sup_{\|\theta-\theta_0\|\ge\epsilon}-\{b(\theta)-b_0\}^\intercal \Delta(\theta)^{-1}\{b(\theta)-b_0\}<0$.
\end{assumption}


\begin{assumption}\label{ass:three}
	For $\theta_0$ defined in Assumption \ref{ass:four}, $\pi(\theta_0)>0$, and $\pi(\cdot)$ is continuous on $\Theta$. For some $p>0$, and all $n$ large enough, $\int_{ \Theta}|v_n^2\Delta_n(\theta)|^{-1/2}\|\theta\|^{p}\pi(\theta)\dt\theta <\infty$.
\end{assumption}

\begin{assumption}\label{ass:propO}
	There exists a function $k:\Theta\rightarrow\mathbb{R}_{+}$ such that: (i) for all $\alpha\in\mathbb{R}^{d}$, $\mathbb{E}\left(\exp \left[\alpha^{\intercal}v_n\left\{S_n(z)-b(\theta)\right\}\right]\right) \leq \exp \left\{\|{\alpha}\|^{2} k(\theta) / 2\right\}$; (ii) there exists a constant $\kappa$  such that $k(\theta)\lesssim \|\theta\|^{\kappa}$; (iii) for all $n$ large enough, $\sup_{\theta\in\Theta}\{\|v_n^2\Delta^{-1}_n(\theta)\|k(\theta)\}<\infty$.
\end{assumption}

 These  assumptions are similar to those used to prove Bernstein--von Mises results in ABC \citep{frazier+mrr18,li+f18b}. In particular, Assumption \ref{ass:one} requires that the observed summaries satisfy a central limit theorem. Assumption \ref{ass:four} ensures that, over $\Theta$, the summaries $S_n(z)$ have a well-behaved limit $b(\theta)$ that is continuous over $\Theta$, can identify $\theta_0$, and whose derivative has full column rank at $\theta_0$. Assumption \ref{ass:four} does not require that $P^{(n)}_0$ corresponds to
$P^{(n)}_{\theta_0}$, so that the model can be misspecified, but instead requires the weaker condition that there exists a unique value $\theta_0\in\Theta$ under which $b(\theta_0)=b_0$, referred to subsequently as the ``true'' parameter value.

{Variants of Assumption~\ref{ass:three} are commonly encountered in the literature on Bayesian asymptotics. In addition to the continuity of $\pi(\theta)$, Assumption~\ref{ass:three} requires the existence of a certain prior moment. This condition is slightly stronger than the prior moment condition needed in the standard case. The need to strengthen this assumption comes from the fact that the matrix $\Delta_n(\theta)$ may be singular far away from $\theta$. As such, in order to ensure the BSL posterior is well-behaved, we require that the prior has thin enough tails in the region where $\Delta_n(\theta)$ is singular, so that the potential singularity of $\Delta_n(\theta)$ does not impact posterior concentration. When $\Delta(\theta)$ in Assumption \ref{ass:two} is positive-definite, uniformly over $\Theta$, this latter condition can be replaced by the standard assumption that $\int_{ \Theta}\|\theta\|^p\pi(\theta)\dt\theta<\infty$ for some $p>0$. }

Assumption \ref{ass:propO} requires that the simulated summaries have a sub-Gaussian tail. Intuitively, this condition requires that the simulated summaries have an exponential moment, and is similar to certain conditions employed by \cite{frazier+mrr18} for ABC.  Without further conditions on the number of model simulations $m$, this assumption seems necessary to ensure that 
the BSL posterior exists, since $\widehat{g}_n(S_n|\theta)$ is defined as an expectation with respect to the distribution of the simulated summaries.

The key difference between the current
 assumptions and those used in the theoretical analysis of ABC is that in BSL 
  the behavior of the quadratic form $\|\Delta_n^{-1/2}(\theta)\{b(\theta)-S_n\}\|^2$  determines the behavior of the synthetic likelihood, and needs to be controlled. 
  Assumption \ref{ass:two}(i) requires that, for $n$ large enough,
   the matrix in this quadratic form is positive-definite for any $\theta$ sufficiently close to $\theta_0$, while Assumption \ref{ass:two}(ii) requires that $\Delta_n(\theta)$ converges uniformly to $\Delta(\theta)$, which is continuous and positive-definite for all $\theta$ sufficiently close to $\theta_0$. Assumption \ref{ass:two}(ii) does not require $\Delta(\theta)$ to be positive-definite uniformly over $\Theta$, and thus it is unnecessary for 
   it to be invertible far from $\theta_0$. This implies that the quadratic form $\|\Delta^{-1/2}(\theta)\{b(\theta)-b_0\}\|^2$ may not be continuous (or finite) uniformly over $\Theta$. In such situations, it is necessary to
    maintain the additional identification assumption given in Assumption \ref{ass:two}(iii). However,  if $\Delta(\theta)$ is continuous over $\Theta$ this identification assumption is automatically satisfied.

Assumptions \ref{ass:one}-\ref{ass:propO} are sufficient to deduce a Bernstein von-Mises result for the BSL posterior. To state this result, define the local parameter $$t:=W_0^{}v_n(\theta-\theta_0)-Z_n,$$where
	$$
	Z_n:=\nabla_{} b\left(\theta_{0}\right)^{\intercal}\Delta(\theta_0)^{-1}v_n\left\{b\left(\theta_{0}\right)-S_ { n }\right\},\;W_0:=\left\{\nabla_{} b\left(\theta_{0}\right)^{\intercal}\Delta(\theta_0)^{-1}\nabla_{} b\left(\theta_{0}\right)\right\}, 
	$$ and denote the BSL posterior for $t$ as 
	$$	\widehat{\pi}(t|S_n):={|W_0^{-1}|}{}\widehat{\pi}\left(\theta_0+W_0^{-1}{t}/v_n+W_0^{-1}Z_n/v_n\large\mid S_n\right)/v_n.
	$$The support of $t$ is denoted by $\mathcal{T}_n:=\{W_0v_n(\theta-\theta_0)-Z_n:\theta\in\Theta\}$, which can be seen as a scaled and shifted translation of $\Theta$.
  The following result states that the total variation distance between $\widehat{\pi}(t|S_n)$, and $N\{t;0,W^{}_0\}$ converges to zero in probability. It also 
  demonstrates that the  covariance of the Gaussian density to which $\widehat{\pi}(t|S_n)$ converges depends on the variance estimator ${\Delta}_n(\theta)$ used in BSL.

\begin{theorem}\label{prop:bvm}
If Assumptions \ref{ass:one}-\ref{ass:propO} are satisfied, and if $m=m(n)\rightarrow\infty$ as $n\rightarrow\infty$, then
$$
\int_{\mathcal{T}_n}\left|\widehat{\pi}(t|S_n)-N\{t;0,W_0^{}\}\right|\dt t=o_{p}(1).
$$
For any $0<\gamma\le2$, if Assumption \ref{ass:three} is satisfied with $p\ge\gamma+\kappa$, then
$$
\int_{\mathcal{T}_n} \|\theta\|^\gamma\left|\widehat{\pi}(t|S_n)-N\{t;0,W_0^{}\}\right|\dt t=o_{p}(1).
$$
\end{theorem}
The second result in Theorem \ref{prop:bvm} demonstrates that, under moment assumptions on the prior, the mean difference between the BSL posterior $\widehat{\pi}(t|S_n)$ and $N\{t;0,W_0\}$ converges to zero in probability. Using this result, we demonstrate that the BSL posterior mean
$\bar{\theta}_n :=\int_{\Theta} \theta \widehat{\pi}(\theta|S_n)\dt\theta$ is asymptotically Gaussian with a covariance matrix that depends on the version of ${\Delta}_n(\theta)$ used in the synthetic likelihood.
\begin{corollary}\label{cor:one}
If the Assumptions in Theorem \ref{prop:bvm} are satisfied, then for $m\rightarrow\infty$ as $n\rightarrow\infty$,
$$
v_{n}(\bar{\theta}_n-\theta_{0}) \Rightarrow N\left[0, W_0^{-1} \left\{ \nabla  b\left(\theta_{0}\right)^{\intercal}\Delta(\theta_0)^{-1}V_0\Delta(\theta_0)^{-1}\nabla  b\left(\theta_{0}\right)\right\}W_0^{-1} \right], \text{ under }P^{(n)}_0.
$$
\end{corollary}

\begin{remark}
	{\normalfont The above results only require weak conditions on the number of simulated datasets, $m$, and  are satisfied for any $m=C\lfloor{n^\gamma}\rfloor$, with $C>0$, $\gamma>0$, and $\lfloor x\rfloor$ denoting the integer floor of $x$. Therefore, Theorem \ref{prop:bvm} and Corollary \ref{cor:one} demonstrate that the choice of $m$ does not strongly impact the resulting inference on $\theta$ and its choice
 should be driven by computational considerations. We note that this requirement is in contrast to ABC, where the choice of tuning parameters, i.e., the tolerance, significantly impacts both the theoretical behavior of ABC and the practical (computing) behavior of ABC algorithms. However,  this lack of dependence on tuning parameters comes at the cost of
  requiring that a version of Assumptions~\ref{ass:two} and \ref{ass:propO} are
   satisfied. ABC requires no condition similar to Assumption \ref{ass:two}, while Assumption \ref{ass:propO} is stronger than the tail conditions on the summaries required  for the ABC posterior to be asymptotically Gaussian. }
\end{remark}

\begin{remark}\label{ref:corr}{\normalfont Theorem \ref{prop:bvm} and Corollary \ref{cor:one} demonstrate the trade-off between using a parsimonious choice for ${\Delta}_n(\theta)$, leading
 to faster computation, and a posterior that correctly quantifies uncertainty. BSL credible sets provide valid uncertainty quantification, in the sense that they have the correct level of asymptotic coverage, when
\begin{flalign*}
\int_{\mathcal{T}_n} t t^{\intercal}\widehat{\pi}(t|S_n)\dt t =&  \nabla  b\left(\theta_{0}\right)^{\intercal}\Delta(\theta^0)^{-1}V_0\Delta(\theta^0)^{-1}\nabla  b\left(\theta_{0}\right)+o_p(1).
\end{flalign*}However, the second part of Theorem \ref{prop:bvm} implies that
$$
\int_{\mathcal{T}_n} t t^{\intercal}\widehat{\pi}(t|S_n)\dt t =  W_0+o_p(1)= \nabla  b\left(\theta_{0}\right)^{\intercal}\Delta(\theta^0)^{-1}\nabla  b\left(\theta_{0}\right)+o_p(1),
$$so that a sufficient condition for the BSL posterior to correctly quantify uncertainty is that
\begin{equation}\label{eq:correct}
{\Delta}(\theta_0)=V_0.
\end{equation}
Satisfying equation \eqref{eq:correct} generally necessitates using
the more computationally intensive variance estimator $\widehat{\Sigma}_n(\theta)$, and that the variance model is ``correctly specified''; here, 
correctly specified means  $\theta_0$ satisfies $b(\theta_0)=b_0$ and  $\theta_0$  also satisfies  equation \eqref{eq:correct}, and where we note that the latter condition {\textit{is not implied by Assumptions \ref{ass:one}-\ref{ass:propO}}}. While a sufficient condition for \eqref{eq:correct} is that $P_\theta^{(n)}=P^{(n)}_0$ for some $\theta_0\in\Theta$, this condition is not necessary in general.
Given the computational costs associated with using $\widehat{\Sigma}_n(\theta)$ when 
the summaries are high-dimensional, Section \ref{adjustments} proposes an adjustment approach to BSL that allows the use of the simpler, possibly misspecified, variance estimator ${\Delta}_n(\theta)$,  but which also yields a posterior that has valid uncertainty quantification.
}		
\end{remark}

\begin{remark}{\normalfont
In contrast to ABC point estimators, Corollary \ref{cor:one} demonstrates that BSL point estimators are generally asymptotically inefficient. 
It is known that $\left\{\nabla b\left(\theta_{0}\right)^{\intercal}V_0^{-1}\nabla  b\left(\theta_{0}\right)\right\}^{-1}$ is the smallest achievable asymptotic variance for any $v_n$-consistent and asymptotically normal estimator of $\theta_0$ based on the parametric class of models $\{P_\theta^{(n)}:\theta\in\Theta\}$ and conditional on the summary statistics $S_n(y)$; see, e.g., \citealp{li+f18a}.
We also have that 
 $$
W_0^{-1} \left\{\nabla  b\left(\theta_{0}\right)^{\intercal}\Delta(\theta^0)^{-1}V_0\Delta(\theta^0)^{-1}\nabla  b\left(\theta_{0}\right)\right\}W_0^{-1}\ge \left\{\nabla  b\left(\theta_{0}\right)^{\intercal}V_0^{-1}\nabla  b\left(\theta_{0}\right)\right\}^{-1};
$$
where
for square matrices $A,B$, $A\ge B$ means that $A-B$ is positive semi-definite. Given this, the BSL posterior mean $\overline \theta_n$ is
 asymptotically efficient only when equation \eqref{eq:correct} is satisfied.
 In this case, BSL simultaneously delivers efficient point estimators and asymptotically correct uncertainty quantification. 
}
\end{remark}

\begin{remark}\normalfont{
		The BSL posterior can be interpreted as a type of quasi-posterior; see, e.g., \citealp{chernozhukov+h03} and \citealp{bissiri2016general}. {However, since the posterior $\widehat{\pi}(\theta|S_n)$ depends on the  ``integrated likelihood'' $\widehat{g}_n(S_n|\theta)$, defined in \eqref{noisySL} and calculated using simulated data,} existing large sample results are not applicable to BSL. }
\end{remark}

\subsection{Computational efficiency}

\cite{li+f18b,li+f18a} discuss the computational efficiency of vanilla and regression-adjusted
ABC algorithms using a rejection sampling method based on a ``good'' proposal density $q_n(\theta)$. They show that regression-adjusted ABC
yields asymptotically correct uncertainty quantification, i.e., credible sets with the correct level of frequentist coverage, and an asymptotically non-zero acceptance rate, while vanilla
ABC can only accomplish one or the other.

This section shows that BSL can deliver correct uncertainty quantification and an asymptotically non-zero acceptance rate, if the number of simulated data sets used in the synthetic likelihood tends to infinity with the sample size. We follow \cite{li+f18b} and consider implementing synthetic likelihood using a rejection sampling algorithm based on the proposal $q_n(\theta)$ 
analogous to the one they consider for ABC.  Following Assumption~\ref{ass:two}(i), there exists a uniform upper bound of the form
$Cv_n^{d_\theta}$ for some $0<C<\infty$ locally in a neighbourhood of $\theta_0$ on $N\left\{S_n;b(\theta),{\Delta}_n(\theta)\right\}$ for $n$ large enough;  an asymptotically valid
rejection sampler then proceeds as follows.
\medskip

\begin{algorithm}[H]
Rejection sampling BSL algorithm
\begin{wideenumerate}
	\item Draw $\theta'\sim q_n(\theta)$
	\item Accept $\theta'$ with probability $(C v_n^{d_\theta})^{-1}\widehat{g}_n(S_n|\theta')=(C v_n^{d_\theta})^{-1}N\left\{S_n;\widehat b_n(\theta'),\Delta_n(\theta')\right\}.$
\end{wideenumerate}
\end{algorithm}

An accepted value from this sampling scheme is a draw from
the density proportional to $q_n(\theta)\widehat{g}_n(S_n|\theta)$. Similarly to the analogous ABC scheme considered in
\cite{li+f18b},  samples from this rejection sampler can be reweighted with importance
weights proportional to $\pi(\theta')/q_n(\theta')$ to recover draws from $\widehat{\pi}(\theta|S_n)\propto \pi(\theta)\widehat{g}_n(S_n|\theta)$.

We choose the proposal density $q_n(\theta)$ to be from the location-scale family $$\mu_n+\sigma_nX,$$ where $X$ is a $d_\theta$-dimensional random variable such that $X\sim q(\cdot)$, $\mathbb{E}_q[X]=0$ and $\mathbb{E}_q[\|X\|^2]<\infty$. The sequences $\mu_n$ and $\sigma_n$ depend on $n$ and satisfy Assumptions \ref{ass:propO} and  \ref{ass:prop}.
\begin{assumption}\label{ass:prop}
	(i) There exists a positive constant $C$, such that $0<\sup_{x}q(x)\leq C<\infty;$ (ii) the sequence $\sigma_n>0$, for all $n\ge1$, satisfies $\sigma_n=o(1)$, and $v_n\sigma_n\rightarrow c_\sigma$, for some positive constant $c_\sigma$; (iii) the sequence $\mu_n$ satisfies  $\sigma^{-1}_n\left(\mu_n-\theta_0\right)=O_p(1)$; (iv) for $h(\theta)=q_n(\theta)/\pi(\theta)$, $\limsup_{n\rightarrow\infty}\int h^2_n(\theta)\pi(\theta|S_n)\dt\theta<\infty$.
\end{assumption}

\begin{remark}
{\normalfont Assumption \ref{ass:prop} formalizes the conditions required of the proposal density and are similar to those required in \cite{li+f18b}. Assumption \ref{ass:prop} is satisfied if the proposal density $q_n(\theta)$ is built from $v_n$-consistent estimators of $\theta_0$, such as those based on pilot runs. }
\end{remark}

The acceptance probability associated with Algorithm 1 is
$$
\widetilde\alpha_n^{}:=({Cv_n^{d_\theta}})^{-1}\int_{\Theta}q_n(\theta)\widehat{g}_n(S_n|\theta)\dt \theta.
$$ We measure the computational efficiency of the rejection
sampling BSL algorithm via the behavior of
$\widetilde\alpha_n^{}$. If $\widetilde\alpha_n^{}$ is asymptotically non-zero, then
 by Corollary 3, and under the restriction in \eqref{eq:correct}, implementing a rejection-based BSL approach can
 yield a posterior that has credible sets with the correct level of frequentist coverage and computational properties that are similar to those of regression-adjusted ABC.

Theorem~\ref{thm:acc} describes the asymptotic behavior of $\widetilde{\alpha}_n$ using the proposal density given in Assumption \ref{ass:prop}. The result uses the following definition: for a random variable $x_n$, we write $x_n=\Xi_p(v_n)$ if there exist constants $0<c\leq C<\infty$ such that $\lim_n\text{pr}\left(c<|x_n/v_n|<C\right)=1$.
\begin{theorem}\label{thm:acc}
	If Assumptions \ref{ass:one}-\ref{ass:prop} are satisfied and if $\int k(\theta)^2\pi(\theta|S_n)\dt\theta<\infty$, then for $m\rightarrow\infty$ as $n\rightarrow\infty$
	$$
	\widetilde\alpha_n=\Xi_p(1)+O_p(1/m).
	$$
\end{theorem}

While Theorem \ref{thm:acc} holds for all choices of $\Delta_n(\theta)$ satisfying Assumption \ref{ass:two}, taking $\Delta_n(\theta)=\Sigma^{}_n(\theta)$ implies that the resulting BSL posterior yields credible sets with the appropriate level of frequentist coverage and that the rejection-based algorithm has a non-negligible acceptance rate asymptotically. Therefore, the result in Theorem \ref{thm:acc} is a BSL version of Theorem 2 in \cite{li+f18b}, demonstrating a similar result, under particular choices of the tolerance sequence, for regression-adjusted ABC.

The example in Section 3 of \citet{price+dln16} compares rejection ABC and a rejection version of synthetic likelihood, where the model is
normal and $\Sigma_n(\theta)$ is constant and does not need to be estimated.  They find that with the prior as the proposal, ABC is more efficient when $d=1$, equally efficient when $d=2$, but less efficient than synthetic likelihood
when $d>2$.  The essence of the example is that the sampling variability in estimating $b_{}(\theta)$ can be equated with the effect of a Gaussian kernel in their
toy normal model for a certain relationship between $\epsilon$ and $m$.  The discussion above suggests that in general models, and
with a good proposal, in large samples the synthetic likelihood is preferable to the vanilla ABC algorithm no matter the dimension of the summary statistic.  However, this
greater computational efficiency is only achieved through the strong tail assumption on the summaries.

\section{Adjustments for misspecification}\label{adjustments}
By~Remark \ref{ref:corr}, if BSL uses a misspecified estimator for the variance for the summaries, in the sense that equation \eqref{eq:correct} does not hold, then the BSL posterior gives invalid uncertainty quantification. This section outlines one approach for adjusting inferences to account for this form of misspecification when Assumption \ref{ass:four} is satisfied, but,
there are other ways to do so.
Suppose $\theta^{q}$, $q=1,\dots, Q$, is an approximate sample from $\widehat{\pi}(\theta|S_n)$, obtained
by  MCMC for example.  Let $\overline{\theta}_n$ denote the synthetic likelihood posterior mean, let $\widetilde{\Gamma}$ denote the synthetic likelihood
posterior covariance, and write $\widehat{\theta}$ and $\widehat{\Gamma}$ for their sample estimates based on $\theta^{q}$, $q=1,\dots, Q$.
Consider the adjusted sample
\begin{align}
  \theta^{A,q} & = \widehat{\theta}+\widehat{\Gamma}\widetilde{\Omega}^{1/2}\widehat{\Gamma}^{-1/2}(\theta^{q}-\widehat{\theta}), \;\;\;\label{adjustedsamp}
\end{align}
$q=1,\dots, Q$,
where $\widetilde{\Omega}$ is an estimate of $\text{var}\left\{\nabla_\theta \log g_n(S_n|\widehat{\theta})\right\}$; the estimation of $\widetilde{\Omega}$ is discussed
below.
We propose using (\ref{adjustedsamp}) as an approximate sample from the posterior, which is similar to the original sample when
the model is correctly specified, but gives asymptotically valid frequentist inference about the pseudo-true parameter value when the model is misspecified.

The motivation for (\ref{adjustedsamp}) is that if $\theta^{q}$ is approximately drawn from the normal distribution $N(\widehat{\theta},\widehat{\Gamma})$,
then $\theta^{A,q}$ is approximately drawn from $N(\widehat{\theta},\widehat{\Gamma}\widetilde{\Omega}\widehat{\Gamma})$.
The results of Corollary 1 imply that
if $\widetilde{\Omega}\approx \text{var}\left\{\nabla_\theta \log g_n(S_n|\theta_0)\right\}$ and $\widehat{\Gamma}$ is approximately
the inverse negative Hessian of $\log g(S_n|\theta)$ at $\theta_0$, then the covariance matrix of the adjusted samples
is approximately that of the sampling distribution of the BSL posterior mean, giving
approximate frequentist validity to
posterior credible intervals based on the adjusted posterior samples.
We now suggest two ways to obtain $\widetilde{\Omega}$.  The first is suitable if
 the model assumed for $y$ is true, but the
 covariance matrix $\lim_n v_n^2 {\Delta}_n(\theta)\neq V_0$, which we refer to as misspecification of the working covariance matrix.
The second way is suitable when the models for both $y$ and the working covariance matrix may be misspecified, but Assumption \ref{ass:four} holds.

\subsection{Estimating $\normalfont{\text{var}}\left\{\nabla_\theta \log g_n(S_n|\theta_0)\right\}$ when the model for $y$ is correct}

\begin{algorithm}\label{alg:one}
Algorithm 1: Estimating $\widetilde{\Omega}$ when the model for $y$ is correct
\begin{enumerate}
\item[]	
\item For $j=1,\dots, J$, draw $S^{(j)}\sim G_n^{\widehat{\theta}}$, where $\widehat{\theta}$ is the estimated synthetic likelihood posterior mean.
\item Approximate $g^{(j)}=\nabla_\theta \log g_n(S^{(j)}|\widehat{\theta})$.
 Section 6.2 discusses the approximation to this gradient as used in the examples.
\item Return
\begin{align*}
  \widetilde{\Omega}=\frac{1}{J-1} \sum_{j=1}^J (g^{(j)}-\bar{g})(g^{(j)}-\bar{g})^{\intercal},
\end{align*}
where $\bar{g}=J^{-1}\sum_{j=1}^J g^{(j)}$.
\end{enumerate}
\end{algorithm}

\subsection{Estimating $\normalfont{\text{var}}\left\{\nabla_\theta \log g_n(S_n|\theta_0)\right\}$ when both the model for $y$ and the covariance matrix may be incorrect}

It may
still be possible estimate $\text{var}\left\{\nabla_\theta \log g_n(S|\theta_0)\right\}$,
 even if the model for $y$ is incorrect.  In particular, if $y_1,\dots, y_n$ are independent,
then we can use the bootstrap to approximate the distribution of $S_n$ at $\theta_0$ and hence estimate
$\text{var}\left\{\nabla_\theta \log g_n(S|\theta_0)\right\}$.  The approximation can be done as in Algorithm 1, but with Step 1 replaced by
\begin{verse}
1.   For $j=1,\dots, J$, sample $y$ with replacement to get a bootstrap sample $y^{(j)}$ with corresponding summary $S^{(j)}$.
\end{verse}
If the data is dependent it may still be possible to use the bootstrap \citep{kriess+p11};
 however the implementation details are model dependent.

\subsection{What the adjustments can and cannot do}

The adjustments suggested above are intended to achieve asymptotically valid frequentist inference when the consistency in \eqref{eq:correct} is not satisfied, i.e., when
$\lim_n {\Delta}_n(\theta)\neq V_0$,
or when the model for $y$ is misspecified, but $S_n$ still satisfies a central limit theorem.
The adjustment will not recover the posterior distribution that is
obtained when the model is correctly specified.  Asymptotically valid frequentist estimation based on the synthetic
likelihood posterior mean for the misspecified synthetic likelihood
is frequentist inference based on a point estimator of $\theta$ that is generally less
efficient than in the correctly specified case.  Matching posterior uncertainty after adjustment to the sampling variability of
such an estimator does not recover the posterior uncertainty from the correctly specified situation.
%

\section{Examples}

\subsection{Toy example}

Suppose that $y_1,\dots, y_n$ are independent observations from a negative binomial distribution $\text{NB}(5,0.5)$ so they  have mean $5$ and variance $10$.
We model the $y_i$ as independent and coming from a $\text{Poisson}(\theta)$ distribution and act as if the likelihood is intractable, basing inference on the sample
mean $\bar{y}$ as the summary statistic $S$.
The pseudo-true parameter value $\theta_0$ is $5$, since this is the parameter value for which the summary statistic mean matches the corresponding mean for the true data generating process.

Under the Poisson model, the synthetic likelihood has $b(\theta)=\theta$ and ${\Delta}_n(\theta)=\theta/n$.  We consider a simulated dataset with $n=20$, and
and deliberately misspecify the variance model in the synthetic likelihood under the Poisson model as ${\Delta}_n(\theta)=\theta/(2n)$.
As noted previously, the deliberate misspecification of $\text{var}(S_n|\theta)$ may be of interest in problems with a high-dimensional $S_n$ as a way of reducing
the number of simulated summaries needed to estimate $\text{var}(S_n|\theta)$ with reasonable precision; for example, we might assume $\text{var}(S_n|\theta)$ is diagonal or based on a factor model.

Figure~\ref{fig1} shows the estimated posterior densities obtained using a number of different
 approaches, when the prior for $\theta$ is $\text{Gamma}(2,0.5)$.
The narrowest green density is obtained from the synthetic likelihood with a misspecified variance.
This density is obtained using 50,000 iterations of a Metropolis-Hastings  MCMC algorithm 
with a normal random walk proposal.
The red density is the exact posterior assuming the Poisson likelihood
is correct, which is $\text{Gamma}(2+n\bar{y},0.5+n)$.  The purple  kernel density estimate based on the adjusted synthetic likelihood samples; it uses the method of Section 5.1 for the adjustment in which the $y$ model is assumed correct but the working covariance matrix is misspecified.  The figure shows that the adjustment gives
a result very close to the exact posterior under an assumed Poisson model.  Finally, the light blue kernel density estimate based on the samples from 
the adjusted synthetic likelihood, uses the method of Section 5.2 based on the bootstrap without assuming that the Poisson model is correct.
This posterior is more dispersed than the one obtained under the Poisson assumption, since the negative binomial generating density
is overdispersed relative to the Poisson, and hence the observed $\bar{y}$ is less informative about the pseudo-true parameter value than implied by the Poisson model.

\begin{figure}
\begin{center}
\begin{tabular}{c}
\includegraphics[width=100mm]{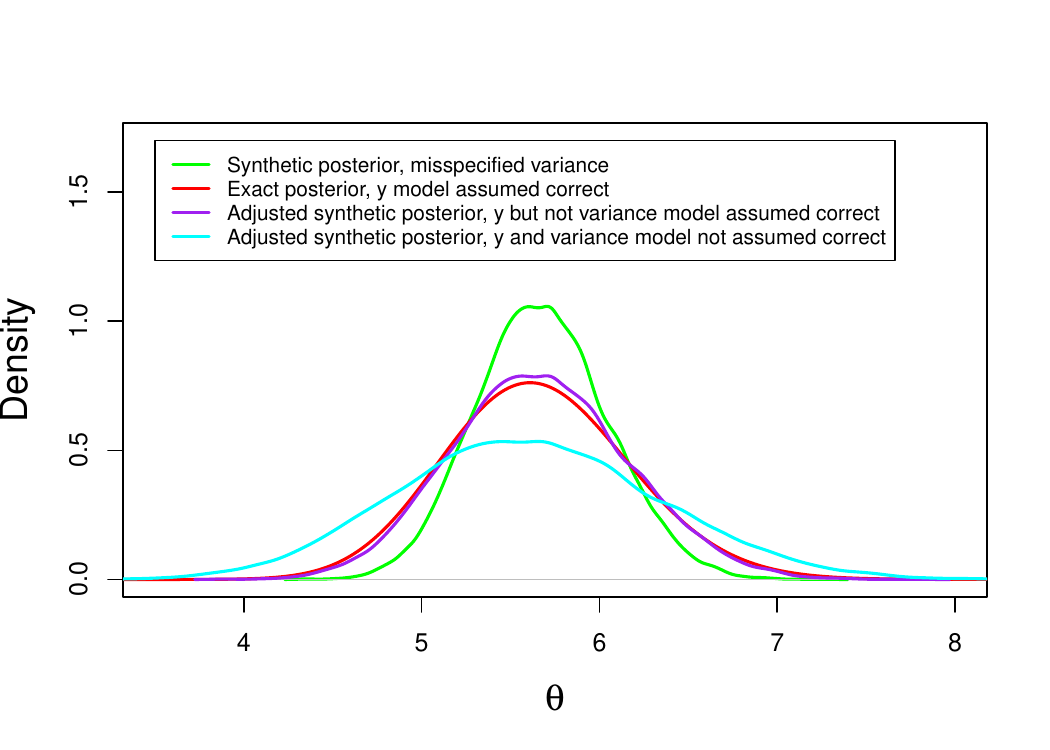}
\end{tabular}
\end{center}
\caption{\label{fig1}
Exact, synthetic and adjusted synthetic posterior densities for the toy example.
}
\end{figure}

\subsection{Examples with a high-dimensional summary statistic}

This section explores the efficacy of the adjustment method when using a misspecified covariance in the presence of a high-dimensional summary statistic $S$. All the examples below 
use the \citet{warton08} shrinkage estimator
 to reduce the number of simulations required to obtain a stable covariance matrix estimate in the synthetic likelihood.
Based on $m$ independent model simulations  the covariance matrix estimate is 
\begin{align}
\widehat{\Sigma}_\gamma & =\widehat{D}^{1/2}\left\{\gamma\widehat{C}+(1-\gamma)I\right\}\widehat{D}^{1/2}, \label{wartonest}
\end{align}
where $\widehat{C}$ is the sample correlation matrix, $\widehat{D}$ is the diagonal matrix of
component sample variances, and  $\gamma \in [0,1]$ is a shrinkage parameter.  The matrix
$\widehat{\Sigma}_\gamma$ is non-singular if $\gamma<1$,
 even if $m$ is less than the dimension of the observations.  This estimator shrinks the sample correlation matrix towards the identity.  When $\gamma = 1$ (resp.\ $\gamma = 0$) there is no shrinkage (resp.\ a diagonal covariance matrix is produced).  We choose $\gamma$ to require only
  1/10 of the simulations required by the  standard synthetic likelihood for Bayesian inference.  We are interested in the shrinkage effect  on the synthetic likelihood approximation and whether our methods can offer a useful correction.  Heavy
shrinkage is used to stabilize covariance estimation in the synthetic likelihood; 
So, the  shrinkage estimator can be thought of as specifying ${\Delta}_n(\theta)$.

To perform the adjustment, it is necessary to approximate the derivative of the synthetic log-likelihood, with shrinkage applied, at a point estimate of the parameter;  we take  this point
as the estimated posterior mean $\widehat{\theta}$ of the BSL approximation.  A computationally efficient approach for estimating these derivatives uses
Gaussian process emulation of the approximate log-likelihood surface based on a pre-computed training sample.  The training sample is constructed around $\widehat{\theta}$,
because this is the only value of $\theta$ for which the approximate derivative is required.   We sample $B$ values using Latin
 hypercube sampling from the hypercube defined by $[\widehat{\theta}_k - \delta_k, \widehat{\theta}_k + \delta_k]$, where $\widehat{\theta}_k$ denotes the $k$th component of $\widehat{\theta}$,
  and take $\delta_k$ as the approximate posterior standard deviation of $\theta_k$; {see \citealp{mckay+bc1979} for details on Latin hypercube sampling}.  Denote the collection of training data as $\mathcal{T} = \{\theta^b,\mu^b,\Sigma_\gamma^b\}_{b=1}^B$, where $\theta^b$ is the $b$th training sample and $\mu^b$ and $\Sigma_\gamma^b$ are the corresponding estimated mean and covariance of the synthetic likelihood from the $m$ model simulations, respectively.  This training sample is stored and recycled for each simulated dataset generated from $\widehat{\theta}$ that needs to be processed in the adjustment method, which is  now 
  described in more detail.

For a simulated statistic $S^{(j)}$ generated from the model at $\widehat{\theta}$, the shrinkage synthetic log-likelihood is rapidly computed at each $\theta^b$ in the training data $\mathcal{T}$ using the pre-stored information, denoted as 
$l^b = l(\theta^b;S^{(j)})$.  A  Gaussian process regression model
 based on the collection $\{\theta^b,l^b\}_{b=1}^B$, is then fitted 
 with $l^b$ as the response and $\theta^b$ as the predictor. We use a zero-mean  Gaussian process with squared exponential covariance function having
  different length scales for different components of $\theta$ and then approximate the gradient of $\log g_n(S^{(j)}|\widehat{\theta})$ by computing the derivative of the smooth predicted mean function of the  Gaussian process at $\widehat{\theta}$.
We can show that this is equivalent to considering the bivariate  Gaussian process of the original process and its derivative, and performing
prediction for the derivative value.  The derivative is estimated using a finite difference approximation because it is simpler than computing the estimate
explicitly. The matrix $\widetilde{\Omega}$ is constructed using $B=200$ training samples and $J=200$ datasets.
Both examples below use 20,000 iterations of MCMC for standard and shrinkage BSL with a multivariate normal random walk proposal.  In each case, the covariance of the random walk was set based on an approximate posterior covariance obtained by pilot  MCMC runs.

\subsection*{Moving average example}

We consider the second order moving average model (MA(2)):
\begin{equation*}
y_t = z_t + \theta_1 z_{t-1} + \theta_2 z_{t-2},
\end{equation*}
for $t = 1,\dots,n$, where $z_t \sim N(0,1)$, $t=-1,\dots,n$, and $n$ is the number of observations in the time series.
{To ensure invertibility of the MA(2) model, the space $\Theta$ is constrained as
 $-1<\theta_2<1,\theta_1+\theta_2>-1,\theta_1-\theta_2<1$ and we specify a uniform prior over this region.}
The density of the observations from an MA(2) model is multivariate normal, with $\mathrm{var}(y_t)=1+\theta_1^2+\theta_2^2$, $\mathrm{cov}(y_t,y_{t-1})=\theta_1+\theta_1\theta_2$, $\mathrm{cov}(y_t,y_{t-2})=\theta_2$, with all other covariances equal to $0$.  The coverage assessment is based on 100 simulated datasets from the model with true parameters $\theta_1 = 0.6$ and $\theta_2 = 0.2$.  Here, we consider a reasonably large sample size of $n=10^4$.

This example uses the first 20 autocovariances as the summary statistic. The autocovariances are a reasonable choice here as they are informative about the parameters and satisfy a central limit theorem \citep{Hannan1976}.

To compare with BSL, we use ABC
with a Gaussian weighting kernel having covariance $\epsilon V$,
where $V$ is a positive-definite matrix.  To favor  the ABC method, 
$V$ is set as the covariance matrix of the summary statistic obtained via many
simulations at the true parameter value.  This ABC
likelihood corresponds to using the Mahalanobis distance function with a Gaussian weighting kernel.  We also consider BSL with a diagonal covariance, and the corresponding adjustment described in Section 4.

To sample from the approximate posterior distributions for each method and dataset, importance sampling with a Gaussian proposal is used with a mean given by the approximate posterior mean and a covariance that is
 twice the approximate posterior covariance.     We treat this as the `good' proposal distribution for posterior inference.  The initial approximations of the (approximate) posteriors are obtained from pilot runs.

For BSL, we use $10,000$ importance samples and consider $m=100$, 200, 500 and $2000$ for estimating the synthetic likelihood.  Table~\ref{tab:ma_coverages_compare} reports 
the mean and minimum effective sample size (ESS)  of the importance sampling approximations \citep{kong92} over the 100 datasets. It shows that for standard BSL with $m=100$ the minimum ESS is small, suggesting this is close to the smallest value of $m$ that can be considered
 to ensure the results are not dominated by Monte Carlo error.  For a given $m$, the ESS values are larger when using a diagonal covariance matrix, demonstrating the computational benefit over estimating a full covariance matrix in standard BSL.  For the BSL adjustment approach, the initial sample before adjustment consists of a re-sample of size 1000 from the relevant diagonal BSL importance sampling approximation to avoid having to work with a weighted sample.

We use 10 million importance samples for ABC-twice as many model simulations compared to BSL with 
$m=500$.  For each dataset, $\epsilon$ is selected so that the ESS is around 1,000, to reduce
$\epsilon$ as much as possible, while  ensuring that the results are robust to Monte Carlo error.   To reduce storage,  a resample of size 1,000 is taken from the ABC importance sampling approximation to produce the final ABC approximation.
We also apply the local regression adjustment of \cite{beaumont+zb02}
to the ABC samples for each dataset.

Table~\ref{tab:ma_coverages_compare} presents the
 estimated marginal coverage rates for $\theta_1$, $\theta_2$ marginally, and the joint coverage for $(\theta_1,\theta_2)$, for nominal coverages of $95\%$, $90\%$ and $80\%$ using kernel density estimates.  The densities are estimated from 1000 samples, performing resampling for the importance sampling approximations when required to avoid dealing with a weighted sample.

 It is evident that standard BSL produces reasonable coverage rates, with some undercoverage at the 80\% nominal rate; $m$ seems to have negligible effect on the  estimated coverage. 
  BSL with a diagonal covariance produces gross overcoverage for $\theta_1$.  Interestingly, despite the overcoverage for $\theta_1$, there is undercoverage at the 95\% and 90\% nominal rates for the joint confidence regions for $\theta_1$ and $\theta_2$,  due to the incorrect estimated dependence structure based on the misspecified covariance.  In contrast, the adjusted BSL results produce accurate coverage rates for the marginals and the joint.

The ABC method produces substantial overcoverage.  ABC with regression adjustment produces more accurate coverage rates, although  some overcoverage remains in general.

\begin{table}
	\centering
	\begin{tabular}{|ccccccc|}
		\hline
		method & $m$ & mean ESS & min ESS &  $95\%$ & $90\%$ & $80\%$ \\
		\hline
		BSL &		100 & 1400 & 21 &		96/97/93 &	91/88/86 &	72/74/73  \\
		BSL &		200 & 3000 & 240 &		95/97/91 &	91/89/88 &	73/78/74  \\
		BSL &		500 & 5000 & 2000 &		95/96/94 &	91/88/88 &	73/74/76  \\
		BSL &		2000 & 6700 & 4900 &		95/97/91 &	89/88/86 &	71/74/75  \\
		BSL diag & 100 & 4200 & 620 & 99/95/89 & 97/88/86 &	95/78/81  \\
		BSL diag &  200 & 5400 & 1500 & 99/95/90 & 98/88/85 &	94/78/81  \\
		BSL diag & 500 & 6500 & 3400 & 99/95/89 & 98/88/87 &	94/78/80  \\
		BSL diag & 2000 & 7200 & 6000 & 99/95/90 & 97/87/87 &	94/78/76  \\
		BSL adj &	100 & - & - & 95/95/94 & 91/92/92 &	80/80/80  \\
		BSL adj &		200 & - & - & 96/95/96 & 91/90/91 &	79/81/77  \\
		BSL adj &		500 & - & - & 94/95/93 & 91/88/86 &	80/80/80  \\
		BSL adj &		2000 & - & - & 95/95/93 & 91/88/85 &	80/78/79  \\
		ABC &		- & - & - &	 98/100/97	 & 96/99/96	 & 89/93/94	  \\
		ABC reg &		- & - & - & 97/97/94	 & 93/96/90	 & 82/84/87	   \\
		\hline
	\end{tabular}
	\caption{Estimated coverage for credible intervals having nominal 95/90/80\% credibility for standard BSL, BSL with a diagonal covariance (BSL diag), BSL diag with an adjustment (BSL adj), ABC and regression adjustment ABC (ABC reg) for $\theta_1/\theta_2/(\theta_1,\theta_2)$.  }
	\label{tab:ma_coverages_compare}
\end{table}

\subsection*{Toad Example}

This example is an individual-based model of a species called Fowler's Toads (\textit{Anaxyrus fowleri}) developed by \citet{Marchand2017}, which was previously analysed by \citet{an+nd18}.  The example is briefly described here;  see \citet{Marchand2017} and \citet{an+nd18} for further details.

The model assumes that a toad hides in its refuge site in the daytime and moves to a randomly chosen foraging place at night.  GPS location data are collected on $n_t$ toads for $n_d$ days, so the matrix of observations $Y$ is $n_d \times n_t$ dimensional.
This example uses both simulated and real data.
The simulated data uses $n_t=66$ and $n_d=63$ and summarize the data by $4$ sets of statistics comprising the relative moving distances for time lags of $1,2,4$ and $8$ days. For instance, $y_1$ consists of the displacement information of lag $1$ day, $y_1 = \{|Y_{i,j}-Y_{i+1,j}| ; 1 \leq i \leq n_{d}-1, 1 \leq j \leq n_t \}$.


Simulating from the model involves two processes. For each toad, we first generate an overnight displacement, $\Delta y$, then mimic the returning behaviour with a simplified model. The overnight displacement is assumed to belong to the L\'evy-alpha stable distribution family, with stability parameter $\alpha$ and scale parameter $\delta$.  With probability $1-p_0$, the toad takes refuge at the location it moved to.  With probability $p_0$, the toad returns to the same refuge site as day $1 \leq i \leq M$ (where $M$ is the number of days the simulation has run for), where $i$ is selected randomly from ${1,2,\dots,M}$ with equal probability.  For the simulated data,  $\theta=(\alpha,\delta,p_0)=(1.7,35,0.6)$, which is a parameter value fitting the real data well, and assume
a uniform prior over $(1,2) \times (0,100) \times (0,0.9)$ for $\theta$.

As in \citet{Marchand2017}, the dataset of displacements is split into two components.  If the absolute value of the displacement is less than 10 metres, it is assumed the toad has returned to its starting location.  For the summary statistic, we consider the number of toads that returned.  For the non-returns (absolute displacement greater than 10 metres), we calculate the log difference between adjacent $p$-quantiles with $p=0,0.1,\ldots,1$ and also the median.  These statistics are computed separately for the four time lags, resulting in a $48$ dimensional statistic.
For standard BSL, $m=500$ simulations are used per  MCMC iteration.  However, with a shrinkage parameter of $\gamma=0.1$,  it was only necessary to use $m=50$ simulations per MCMC iteration.  For the simulated data,   the  MCMC acceptance rates are 16\% and 21\% for standard and shrinkage BSL, respectively.   For the real data, the acceptance rates are both roughly 24\%.

Figure \ref{fig:toad_results} summarizes the results for the simulated data and shows that the
shrinkage BSL posterior underestimates the variance and has the wrong
dependence structure compared to the standard BSL posterior.
The adjusted posterior produces uncertainty quantification that is closer to the standard BSL procedure,
although its larger variances indicate that there is a loss in efficiency in using frequentist inference based on
the shrinkage BSL point estimate.  The results for the real data in Figure \ref{fig:toad_results_realdata} are qualitatively similar.  There is less difference in the posterior means between the standard and shrinkage BSL methods for the real data compared to the simulated data, and generally less variance inflation in the adjusted results for the real data compared to the simulated data.

\begin{figure}
	\centering
	\includegraphics[width=150mm]{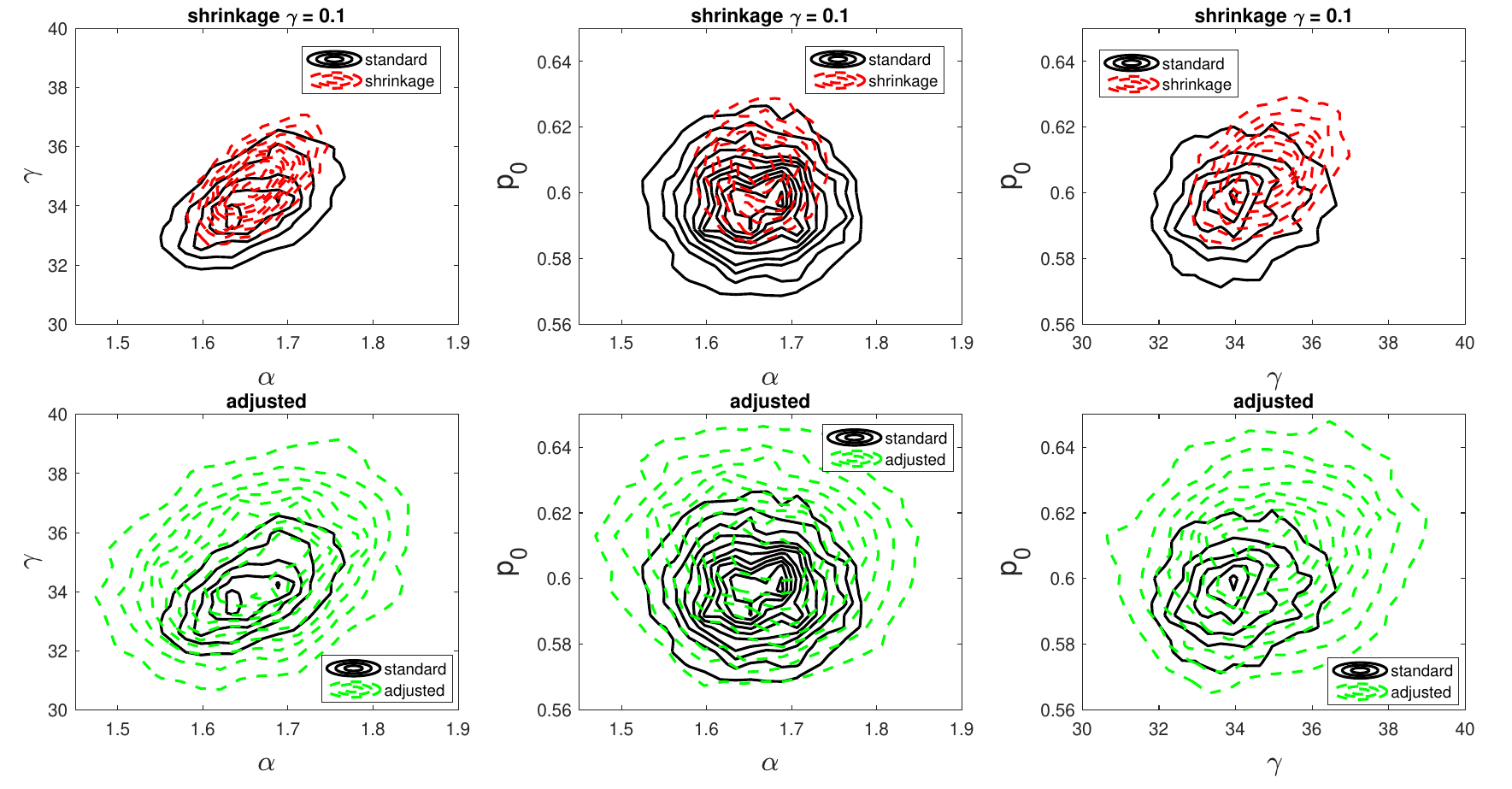}
	\caption{\label{fig:toad_results}
		Adjustment results for the toad example based on the simulated data.  The panels in the top row are  bivariate contour plots of the standard and shrinkage BSL posteriors.  The panels in the bottom row are bivariate contour plots of the standard and adjusted BSL posteriors.
	}
\end{figure}

\begin{figure}
	\centering
	\includegraphics[width=150mm]{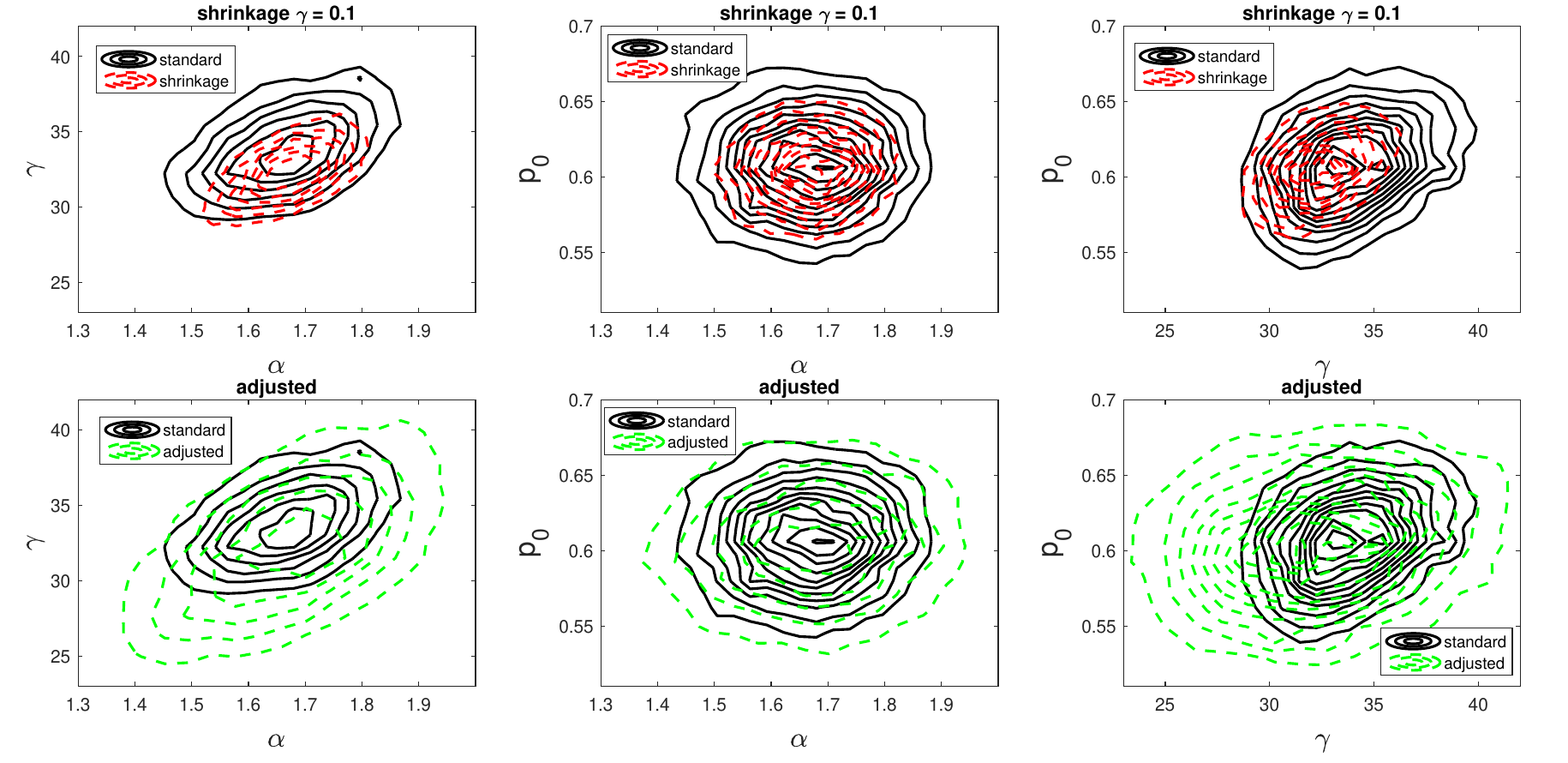}
	\caption{\label{fig:toad_results_realdata}
		Adjustment results for the toad example based on the real data.  The top row panels are  bivariate contour plots of the standard and shrinkage BSL posteriors.  The bottom row panels are bivariate contour plots of the standard and adjusted BSL posteriors.
	}
\end{figure}

\section{Discussion}

Our article examines the asymptotic behaviour of Bayesian inference using the synthetic likelihood when the summary statistic satisfies a central limit theorem.  The synthetic likelihood asymptotically quantifies uncertainty similarly to ABC methods under appropriate algorithmic settings and assumptions leading to correct uncertainty quantification.  We also examine
the effect of estimating the mean and covariance matrix in synthetic likelihood algorithms, as well as the computational efficiency of similar versions of rejection and importance sampling algorithms for BSL
and ABC.  BSL is more efficient than vanilla ABC, 
 and behaves similarly to regression-adjusted ABC.

{Adjustments are also discussed for a  misspecified synthetic likelihood covariance
of the synthetic likelihood.}  
These adjustments may also
be useful when the model for $y$ is misspecified, and inference on the pseudo-true parameter
 is of interest.
{Our adjustment methods do not help correct inference in the case where the summary statistics are not normal.}  Some approaches
consider more complex parametric models than the normal for addressing this issue, and the asymptotic framework
developed here could be adapted to other parametric model approximations for the summaries.
These extensions are left to future work.

Although our adjustments could be useful when the model for $y$ is misspecified, it is helpful to distinguish
different types of misspecification.
Model incompatibility is said to occur when it is impossible to recover the observed summary statistic
for any $\theta$, but we do not investigate
 the behaviour of synthetic likelihood in detail in this case.
\citet{Frazier2019} and \citet{frazier+rr17} demonstrate that standard BSL and ABC can both perform poorly under incompatibility. \citet{Frazier2019} propose some extensions to BSL allowing greater robustness and computational efficiency
in this setting. More research is  needed to compare BSL and ABC when
model incompatibility occurs.

\section*{Acknowledgments}
David Frazier was supported by the Australian Research Council's Discovery Early Career Researcher Award funding scheme (DE200101070). David Nott was supported by a Singapore Ministry of Education Academic Research Fund Tier 1
grant and is affiliated with the Operations Research and Analytics Research cluster at the National University of Singapore.    Christopher Drovandi was supported by an Australian Research Council Discovery Project (DP200102101).    Robert Kohn was partially supported by the Center of Excellence grant CE140100049 and Robert Kohn,
Christopher Drovandi and David Frazier are affiliated with the Australian Centre of Excellence for Mathematical and Statistical Frontiers.  We thank Ziwen An for preparing computer code for the toad example.

\bibliographystyle{chicago}
\bibliography{references}

\begin{thebibliography}{}

\bibitem[\protect\citeauthoryear{An, Nott, and Drovandi}{An
  et~al.}{2020}]{an+nd18}
An, Z., D.~J. Nott, and C.~Drovandi (2020).
\newblock Robust {B}ayesian synthetic likelihood via a semi-parametric
  approach.
\newblock {\em Statistics and Computing\/}~{\em 30}, 543--557.

\bibitem[\protect\citeauthoryear{An, South, Drovandi, and Nott}{An
  et~al.}{2019}]{An2016}
An, Z., L.~F. South, C.~C. Drovandi, and D.~J. Nott (2019).
\newblock Accelerating {B}ayesian synthetic likelihood with the graphical
  lasso.
\newblock {\em Journal of Computational and Graphical Statistics\/}~{\em
  28\/}(2), 471--475.

\bibitem[\protect\citeauthoryear{Beaumont, Zhang, and Balding}{Beaumont
  et~al.}{2002}]{beaumont+zb02}
Beaumont, M.~A., W.~Zhang, and D.~J. Balding (2002).
\newblock Approximate {Bayesian} computation in population genetics.
\newblock {\em Genetics\/}~{\em 162}, 2025--2035.

\bibitem[\protect\citeauthoryear{Bissiri, Holmes, and Walker}{Bissiri
  et~al.}{2016}]{bissiri2016general}
Bissiri, P.~G., C.~C. Holmes, and S.~G. Walker (2016).
\newblock A general framework for updating belief distributions.
\newblock {\em Journal of the Royal Statistical Society. Series B, Statistical
  methodology\/}~{\em 78\/}(5), 1103.

\bibitem[\protect\citeauthoryear{Chaudhuri, Ghosh, Nott, and Pham}{Chaudhuri
  et~al.}{2020}]{chaudhuri+gnp18}
Chaudhuri, S., S.~Ghosh, D.~J. Nott, and K.~C. Pham (2020).
\newblock On a variational approximation based empirical likelihood {ABC}
  method.
\newblock arXiv:2011.07721.

\bibitem[\protect\citeauthoryear{Chernozhukov and Hong}{Chernozhukov and
  Hong}{2003}]{chernozhukov+h03}
Chernozhukov, V. and H.~Hong (2003).
\newblock An {MCMC} approach to classical estimation.
\newblock {\em Journal of Econometrics\/}~{\em 115\/}(2), 293 -- 346.

\bibitem[\protect\citeauthoryear{Deligiannidis, Doucet, and Pitt}{Deligiannidis
  et~al.}{2018}]{deligiannidis2018correlated}
Deligiannidis, G., A.~Doucet, and M.~K. Pitt (2018).
\newblock The correlated pseudo-marginal method.
\newblock {\em Journal of the Royal Statistical Society: Series B (Statistical
  Methodology)\/}~{\em 80\/}(5), 839--870.

\bibitem[\protect\citeauthoryear{Doucet, Pitt, Deligiannidis, and Kohn}{Doucet
  et~al.}{2015}]{doucet+pdk15}
Doucet, A., M.~K. Pitt, G.~Deligiannidis, and R.~Kohn (2015).
\newblock {Efficient implementation of Markov chain Monte Carlo when using an
  unbiased likelihood estimator}.
\newblock {\em Biometrika\/}~{\em 102\/}(2), 295--313.

\bibitem[\protect\citeauthoryear{Drovandi, Pettitt, and Lee}{Drovandi
  et~al.}{2015}]{drovandi+pl15}
Drovandi, C.~C., A.~N. Pettitt, and A.~Lee (2015).
\newblock Bayesian indirect inference using a parametric auxiliary model.
\newblock {\em Statistical Science\/}~{\em 30\/}(1), 72--95.

\bibitem[\protect\citeauthoryear{Everitt}{Everitt}{2017}]{everitt17}
Everitt, R.~G. (2017).
\newblock Boostrapped synthetic likelihood.
\newblock arXiv:1711.05825.

\bibitem[\protect\citeauthoryear{Fasiolo, Wood, Hartig, and Bravington}{Fasiolo
  et~al.}{2018}]{Fasiolo2016}
Fasiolo, M., S.~N. Wood, F.~Hartig, and M.~V. Bravington (2018).
\newblock An extended empirical saddlepoint approximation for intractable
  likelihoods.
\newblock {\em Electronic Journal of Statistics\/}~{\em 12\/}(1), 1544--1578.

\bibitem[\protect\citeauthoryear{Forneron and Ng}{Forneron and
  Ng}{2018}]{forneron+n18}
Forneron, J.-J. and S.~Ng (2018).
\newblock {The ABC of simulation estimation with auxiliary statistics}.
\newblock {\em Journal of Econometrics\/}~{\em 205\/}(1), 112--139.

\bibitem[\protect\citeauthoryear{Frazier and Drovandi}{Frazier and
  Drovandi}{2019}]{Frazier2019}
Frazier, D.~T. and C.~Drovandi (2019).
\newblock Robust approximate {B}ayesian inference with synthetic likelihood.
\newblock {\em arXiv preprint arXiv:1904.04551\/}.

\bibitem[\protect\citeauthoryear{Frazier, Martin, Robert, and Rousseau}{Frazier
  et~al.}{2018}]{frazier+mrr18}
Frazier, D.~T., G.~M. Martin, C.~P. Robert, and J.~Rousseau (2018).
\newblock Asymptotic properties of approximate {B}ayesian computation.
\newblock {\em Biometrika\/}~{\em 105\/}(3), 593--607.

\bibitem[\protect\citeauthoryear{Frazier, Robert, and Rousseau}{Frazier
  et~al.}{2020}]{frazier+rr17}
Frazier, D.~T., C.~P. Robert, and J.~Rousseau (2020).
\newblock Model misspecification in approximate {B}ayesian computation:
  consequences and diagnostics.
\newblock {\em Journal of the Royal Statistical Society: Series B (Statistical
  Methodology)\/}.

\bibitem[\protect\citeauthoryear{Gutmann and Corander}{Gutmann and
  Corander}{2016}]{gutmann+c15}
Gutmann, M.~U. and J.~Corander (2016).
\newblock Bayesian optimization for likelihood-free inference of
  simulator-based statistical models.
\newblock {\em Journal of Machine Learning Research\/}~{\em 17\/}(125), 1--47.

\bibitem[\protect\citeauthoryear{Hannan}{Hannan}{1976}]{Hannan1976}
Hannan, E.~J. (1976).
\newblock The asymptotic distribution of serial covariances.
\newblock {\em The Annals of Statistics\/}~{\em 4\/}(2), 396--399.

\bibitem[\protect\citeauthoryear{Hsu, Kakade, Zhang, et~al.}{Hsu
  et~al.}{2012}]{hsu2012tail}
Hsu, D., S.~Kakade, T.~Zhang, et~al. (2012).
\newblock A tail inequality for quadratic forms of subgaussian random vectors.
\newblock {\em Electronic Communications in Probability\/}~{\em 17}.

\bibitem[\protect\citeauthoryear{Kong}{Kong}{1992}]{kong92}
Kong, A. (1992).
\newblock A note on importance sampling using standardized weights.
\newblock Chicago Dept. of Statistics Tech. Rep 348.

\bibitem[\protect\citeauthoryear{Kreiss and Paparoditis}{Kreiss and
  Paparoditis}{2011}]{kriess+p11}
Kreiss, J.-P. and E.~Paparoditis (2011).
\newblock Bootstrap methods for dependent data: A review.
\newblock {\em Journal of the Korean Statistical Society\/}~{\em 40\/}(4), 357
  -- 378.

\bibitem[\protect\citeauthoryear{Lehmann and Casella}{Lehmann and
  Casella}{1998}]{lehmann2006theory}
Lehmann, E.~L. and G.~Casella (1998).
\newblock {\em Theory of point estimation}.
\newblock Springer Science \& Business Media.

\bibitem[\protect\citeauthoryear{Li and Fearnhead}{Li and
  Fearnhead}{2018a}]{li+f18b}
Li, W. and P.~Fearnhead (2018a).
\newblock {Convergence of regression-adjusted approximate Bayesian
  computation}.
\newblock {\em Biometrika\/}~{\em 105\/}(2), 301--318.

\bibitem[\protect\citeauthoryear{Li and Fearnhead}{Li and
  Fearnhead}{2018b}]{li+f18a}
Li, W. and P.~Fearnhead (2018b).
\newblock {On the asymptotic efficiency of approximate Bayesian computation
  estimators}.
\newblock {\em Biometrika\/}~{\em 105\/}(2), 285--299.

\bibitem[\protect\citeauthoryear{Marchand, Boenke, and Green}{Marchand
  et~al.}{2017}]{Marchand2017}
Marchand, P., M.~Boenke, and D.~M. Green (2017).
\newblock A stochastic movement model reproduces patterns of site fidelity and
  long-distance dispersal in a population of {F}owler's toads ({A}naxyrus
  fowleri).
\newblock {\em Ecological Modelling\/}~{\em 360}, 63 -- 69.

\bibitem[\protect\citeauthoryear{McKay, Beckman, and Conover}{McKay
  et~al.}{1979}]{mckay+bc1979}
McKay, M., R.~Beckman, and W.~Conover (1979).
\newblock Comparison of three methods for selecting values of input variables
  in the analysis of output from a computer code.
\newblock {\em Technometrics\/}~{\em 21\/}(2), 239--245.

\bibitem[\protect\citeauthoryear{Meeds and Welling}{Meeds and
  Welling}{2014}]{meeds+w14}
Meeds, E. and M.~Welling (2014).
\newblock {GPS-ABC: G}aussian process surrogate approximate {B}ayesian
  computation.
\newblock In {\em Proceedings of the Thirtieth Conference on Uncertainty in
  Artificial Intelligence}, UAI'14, Arlington, VA, pp.\  593--602. AUAI Press.

\bibitem[\protect\citeauthoryear{Mengersen, Pudlo, and Robert}{Mengersen
  et~al.}{2013}]{mengersen+pr13}
Mengersen, K.~L., P.~Pudlo, and C.~P. Robert (2013).
\newblock Bayesian computation via empirical likelihood.
\newblock {\em Proceedings of the National Academy of Sciences\/}~{\em
  110\/}(4), 1321--1326.

\bibitem[\protect\citeauthoryear{M\"{u}ller}{M\"{u}ller}{2013}]{muller13}
M\"{u}ller, U.~K. (2013).
\newblock Risk of {B}ayesian inference in misspecified models, and the sandwich
  covariance matrix.
\newblock {\em Econometrica\/}~{\em 81\/}(5), 1805--1849.

\bibitem[\protect\citeauthoryear{Ong, Nott, Tran, Sisson, and Drovandi}{Ong
  et~al.}{2018a}]{ong+ntsd16}
Ong, V. M.-H., D.~J. Nott, M.-N. Tran, S.~Sisson, and C.~Drovandi (2018a).
\newblock Variational {B}ayes with synthetic likelihood.
\newblock {\em Statistics and Computing\/}~{\em 28\/}(4), 971--988.

\bibitem[\protect\citeauthoryear{Ong, Nott, Tran, Sisson, and Drovandi}{Ong
  et~al.}{2018b}]{ong+ntsd18}
Ong, V. M.-H., D.~J. Nott, M.-N. Tran, S.~A. Sisson, and C.~C. Drovandi
  (2018b).
\newblock Likelihood-free inference in high dimensions with synthetic
  likelihood.
\newblock {\em Computational Statistics and Data Analysis\/}~{\em 128},
  271--291.

\bibitem[\protect\citeauthoryear{Pitt, Silva, Giordani, and Kohn}{Pitt
  et~al.}{2012}]{pitt+sgk12}
Pitt, M.~K., R.~d.~S. Silva, P.~Giordani, and R.~Kohn (2012).
\newblock {On some properties of Markov chain Monte Carlo simulation methods
  based on the particle filter}.
\newblock {\em Journal of Econometrics\/}~{\em 171\/}(2), 134--151.

\bibitem[\protect\citeauthoryear{Price, Drovandi, Lee, and Nott}{Price
  et~al.}{2018}]{price+dln16}
Price, L.~F., C.~C. Drovandi, A.~C. Lee, and D.~J. Nott (2018).
\newblock Bayesian synthetic likelihood.
\newblock {\em Journal of Computational and Graphical Statistics\/}~{\em
  27\/}(1), 1--11.

\bibitem[\protect\citeauthoryear{Priddle, Sisson, Frazier, and
  Drovandi}{Priddle et~al.}{2019}]{priddle+sfd19}
Priddle, J.~W., S.~A. Sisson, D.~T. Frazier, and C.~Drovandi (2019).
\newblock Efficient {B}ayesian synthetic likelihood with whitening
  transformations.
\newblock arXiv:1909.04857.

\bibitem[\protect\citeauthoryear{Sherlock, Thiery, Roberts, and
  Rosenthal}{Sherlock et~al.}{2015}]{sherlock2015}
Sherlock, C., A.~H. Thiery, G.~O. Roberts, and J.~S. Rosenthal (2015, 02).
\newblock On the efficiency of pseudo-marginal random walk metropolis
  algorithms.
\newblock {\em Ann. Statist.\/}~{\em 43\/}(1), 238--275.

\bibitem[\protect\citeauthoryear{Sisson, Fan, and Beaumont}{Sisson
  et~al.}{2018}]{sisson2018}
Sisson, S.~A., Y.~Fan, and M.~Beaumont (2018).
\newblock {\em Handbook of Approximate Bayesian Computation}.
\newblock Chapman and Hall/CRC.

\bibitem[\protect\citeauthoryear{Thomas, Dutta, Corander, Kaski, and
  Gutmann}{Thomas et~al.}{2021}]{Dutta2016}
Thomas, O., R.~Dutta, J.~Corander, S.~Kaski, and M.~U. Gutmann (2021).
\newblock Likelihood-free inference by ratio estimation.
\newblock {\em Bayesian Analysis\/}~(To Appear).

\bibitem[\protect\citeauthoryear{Tran, Kohn, Quiroz, and Villani}{Tran
  et~al.}{2016}]{tran2016block}
Tran, M.-N., R.~Kohn, M.~Quiroz, and M.~Villani (2016).
\newblock The block pseudo-marginal sampler.
\newblock {\em arXiv preprint arXiv:1603.02485\/}.

\bibitem[\protect\citeauthoryear{Warton}{Warton}{2008}]{warton08}
Warton, D.~I. (2008).
\newblock Penalized normal likelihood and ridge regularization of correlation
  and covariance matrices.
\newblock {\em Journal of the American Statistical Association\/}~{\em 103},
  340--349.

\bibitem[\protect\citeauthoryear{Wilkinson}{Wilkinson}{2014}]{wilkinson14}
Wilkinson, R. (2014).
\newblock Accelerating {ABC} methods using {G}aussian processes.
\newblock {\em Journal of Machine Learning Research\/}~{\em 33}, 1015--1023.

\bibitem[\protect\citeauthoryear{Wood}{Wood}{2010}]{wood10}
Wood, S.~N. (2010).
\newblock Statistical inference for noisy nonlinear ecological dynamic systems.
\newblock {\em Nature\/}~{\em 466}, 1102--1107.

\end{thebibliography}



\appendix
\section{Proofs and Lemmas}
\subsection{Proofs of the main results}
\begin{proof}[Proof of Theorem \ref{prop:bvm}]
We only prove the second result in Theorem \ref{prop:bvm}, the first result then follows by taking $\gamma=0$. Upper bound the integral in question as
\begin{flalign}
\int_{\mathcal{T}_n} \|t\|^{\gamma}|\widehat{\pi}(t|S_n)-N\{t;0,W_0\}|\dt t\le& \int_{\mathcal{T}_n} \|t\|^{\gamma}|{\pi}(t|S_n)-N\{t;0,W_0\}|\dt t\nonumber\\&+\int_{\mathcal{T}_n} \|t\|^\gamma |\widehat{\pi}(t|S_n)-{\pi}(t|S_n)|\dt t, \label{eq:decomp}
\end{flalign}and the stated result follows if both terms in \eqref{eq:decomp} are $o_p(1)$. The first term on the RHS of \eqref{eq:decomp} is $o_p(1)$ by Lemma \ref{prop:bvm1}; we now show that the second term is $o_p(1)$.

Define $M^{}_n(\theta):=[v_n^2{\Delta}_n(\theta)]^{-1}$, $Q_n(\theta):=-v_n^2\{b(\theta)-S_n\}^{\intercal}M_n(\theta)\{b(\theta)-S_n\}/2$, and $$\widehat Q_{n}(\theta):=-v_n^2\left\{\widehat b_{n}(\theta)-S_n\right\}^{\intercal}M_n(\theta)\left\{\widehat b_{n}(\theta)-S_n\right\}/2.$$ We first demonstrate that, uniformly over $\Theta$,
\begin{flalign}\label{eq:app1}
\mathbb{E}\left[\exp\left\{\widehat Q_{n}(\theta)\right\}\mid\theta,S_n\right]= \exp\left\{Q_n(\theta)\right\}\left\{1+O\left(\frac{k(\theta)}{m}\right)\right\}.
\end{flalign}
Using properties of quadratic forms, and Assumption \ref{ass:propO},
\begin{flalign*}
\mathbb{E}\left[\left\|M^{1/2}_n(\theta)v_n\left\{\widehat{b}_n(\theta)-S_n\right\}\right\|^2\mid\theta,S_n\right]=&\text{Tr}\left\{M_n(\theta)\cdot\text{Cov}\left[v_n\left\{\widehat{b}_n(\theta)-S_n\right\}\right]\right\}+\mu_n(\theta)^\intercal M_n(\theta) \mu_n(\theta)\\&\leq \text{Tr}[M_n(\theta)] k(\theta)/m+\mu_n(\theta)^\intercal M_n(\theta)\mu_n(\theta),
\end{flalign*}
where $$\mu_n(\theta)=v_n\mathbb{E}[S_n(z^i)-S_n|S_n,\theta]=v_n\{b(\theta)-S_n\}.$$ Apply Lemma 
\ref{lem:propO} with $A=M^{1/2}_n(\theta)$, $x=v_n\{\widehat{b}_n(\theta)-S_n\}$, and 
$M=M_n(\theta)$, which is valid for $\eta$ satisfying $$0\leq \eta<{1}\big{/}{\left[2\frac{k(\theta)}{m}\|M_n(\theta)\|\right]}.$$ However, by Assumption~\ref{ass:propO}(ii), for any $\theta\in\Theta$, $\|M_n(\theta)\|k(\theta)/m=o(1)$ as $n\rightarrow\infty$. Therefore,
for $n$ large enough and uniformly over $\Theta$ , we take $\eta = 1$, without loss of generality.  
Applying Lemma~\ref{lem:propO}, with $\eta=1$, yields
\begin{flalign}\label{eq:logeq}
\log\left\{\mathbb{E}[\exp(\|Ax\|^2)]\right\}&\leq \text{Tr}[M_n(\theta)]k(\theta)/m+\frac{\|M^{1/2}_n
(\theta)\mu_n(\theta)\|^2}{[1+o(1)]}+O\left(\frac{\text{Tr}[M_n(\theta)^2]k^2(\theta)/m^2}
{1+o(1)}\right).
\end{flalign} One half of the numerator of the second term in the above equation is equivalent to $$\mu_n(\theta)^\intercal M^{}_n(\theta)\mu_n(\theta)/2=v_n^2\{b(\theta)-S_n\}^\intercal M_n(\theta)\{b(\theta)-S_n\}/2= -Q_n(\theta).$$ Therefore, from equation \eqref{eq:logeq},
\begin{flalign*}
\mathbb{E}\left[\exp\{\widehat{Q}_n(\theta)\}\mid\theta,S_n\right]&=
\exp(-\|M_n^{1/2}(\theta)\mu_n(\theta)\|^2/2)\exp\left[O\left\{\text{Tr}\left[M_n(\theta)
\right]k(\theta)/m\right\}\right]\\&=\exp\left\{Q_n(\theta)\right\}O\left\{\text{Tr}
\left[M_n(\theta)\right]k(\theta)/m\right\}\\&\leq\exp\left\{Q_n(\theta)\right\}
\left\{1+O(k(\theta)/m)\right\}
\end{flalign*}

From equation  \eqref{eq:app1} and the definitions of $\widehat{g}_n(\theta|S_n)$ and $g_n(\theta|S_n)$, 
\begin{flalign}\label{eq:app3}
|\widehat{g}_n(S_n|\theta)-g_n(S_n|\theta)|\leq g_n(S_n|\theta)\left[O\left\{{k(\theta)}/{m}\right\}\right],
\end{flalign}so that
\begin{flalign*}
\left|\int_{\Theta} \widehat{g}_n(S_n|\theta)\pi(\theta)\dt \theta-\int_{\Theta} {g}_n(S_n|\theta)\pi(\theta)\dt \theta\right|&\leq \int_\Theta|\widehat{g}_n(S_n|\theta)-{g}_n(S_n|\theta)|\pi(\theta)\dt \theta\\&\lesssim \frac{1}{m}\int_{\Theta} k(\theta)g_n(S_n|\theta)\pi(\theta)\dt \theta\\&= \frac{1}{m}\int_{\Theta} g_n(S_n|\theta)\pi(\theta)\dt \theta\int_{\Theta} k(\theta)\pi(\theta|S_n)\dt \theta,
\end{flalign*}where the second line follows from equation \eqref{eq:app3}, and the equality from reorganizing terms.

The proof of Lemma \ref{prop:bvm1} demonstrates that $\int_\Theta g_n(S_n|\theta)\pi(\theta)\dt \theta<\infty$ for all $n$ large enough; hence,  
\begin{flalign}
\left|\int_{\Theta} \|\theta\|^{\gamma}\widehat{g}_n(S_n|\theta)\pi(\theta)\dt \theta-\int_{\Theta} \|\theta\|^{\gamma}{g}_n(S_n|\theta)\pi(\theta)\dt \theta\right|&\lesssim \frac{1}{m}\int_{\Theta}\|\theta\|^{\gamma}k(\theta)\pi(\theta|S_n)\dt \theta\nonumber\\&\lesssim \frac{1}{m}\int_{\Theta}\|\theta\|^{\xi}\pi(\theta|S_n)\dt \theta, \label{eq:neweq}
\end{flalign}
where $\xi=\gamma+\kappa$, with $\kappa$ as in Assumption~\ref{ass:propO}(ii).
Consider the term $\int_\Theta\|\theta\|^\xi\pi(\theta|S_n)\dt\theta$. Recall that $t:=v_nW_0^{}(\theta-\theta_0)-Z_n$, and we obtain
\begin{flalign*}
\|\theta\|^{\xi}&= \|W_0^{-1}t/v_n+\theta_0+W_0^{-1}Z_n/v_n\|^{\xi}\lesssim{v_n^{-\xi}}
\|t\|^{\xi}+\|\{b(\theta_0)-S_n\}+\theta_0\|^{\xi}.
\end{flalign*}
Applying the change of variables $\theta\mapsto t$ and the above inequality yields
\begin{flalign}
\int_{}\|\theta\|^{\xi}\pi(\theta|S_n)\dt \theta&\lesssim {v_n^{-\xi}}\int\|t\|^{\xi}\pi(t|S_n) \dt t+\|\{b(\theta_0)-S_n\}+\theta_0\|^{\xi}\label{eq:app4}.
\end{flalign}
Now, 
\begin{flalign*}
\int\|t\|^{\xi}\pi(t|S_n) \dt t&\leq \int\|t\|^{\xi}|\pi(t|S_n)-N\{t;0,W_0\}|\dt t+\int\|t\|^{\xi}N\{t;0,W_0\}\dt t
\end{flalign*}
The first term in the above equation is $o_p(1)$ by Lemma \ref{prop:bvm1} under Assumption \ref{ass:three} with $p\ge\xi$, and
the second term is finite due to Gaussianity; hence, 
\begin{flalign}
\int\|t\|^{\xi}\pi(t|S_n) \dt t&=o_p(1)+C\label{eq:app5}.
\end{flalign}
Using  equation \eqref{eq:app5} in equation \eqref{eq:app4}, and the fact that, by Assumption~\ref{ass:one}, $\|S_n- b(\theta_0)\|=o_p(1)$,
\begin{flalign}
\int_{\Theta}\|\theta\|^{\xi}\pi(\theta|S_n)\dt \theta&\leq C/v_n^{\xi}+o_p(1/v_n^{\xi})+{\|\theta_0+o_p(1)\|^{\xi}}{}.\label{eq:new2}
\end{flalign}Applying equation \eqref{eq:new2} into the RHS of equation \eqref{eq:neweq} then yields, 
\begin{flalign}
\left|\int_{\Theta}\|\theta\|^{\gamma} \widehat{g}_n(S_n|\theta)\pi(\theta)\dt \theta-\int_{\Theta}\|\theta\|^{\gamma} {g}_n(S_n|\theta)\pi(\theta)\dt \theta\right|&\lesssim \frac{1}{m}\int_{\Theta}\|\theta\|^{\xi}\pi(\theta|S_n)\dt \theta= O_p(1/m).\label{eq:result1}
\end{flalign}

It then follows from equation \eqref{eq:result1} that
\begin{flalign}\label{eq:result2}
\frac{\left|\int \widehat{g}_n(S_n|\theta)\pi(\theta)\dt \theta-\int {g}_n(S_n|\theta)\pi(\theta)\dt \theta\right|}{\int g_n(S_n|\theta)\pi(\theta)\dt\theta}= O_p(1/m),
\end{flalign} and so
\begin{flalign*}
\frac{\int_{\Theta} \widehat{g}_n(S_n|\theta)\pi(\theta)\dt \theta}{\int_{\Theta} g_n(S_n|\theta)\pi(\theta)\dt\theta}=1+O_p(1/m);\text{ }\frac{\int_{\Theta} g_n(S_n|\theta)\pi(\theta)\dt\theta}{\int_{\Theta} \widehat{g}_n(S_n|\theta)\pi(\theta)\dt \theta}=1+O_p(1/m).
\end{flalign*}
Write  $\{\widehat{\pi}_{}(\theta|S_n)-\pi_{}(\theta|S_n)\}$ as 
\begin{flalign*}
\{\widehat{\pi}_{}(\theta|S_n)-\pi_{}(\theta|S_n)\}=& \frac{\widehat{g}_n(S_n|\theta)\pi(\theta) }{\int_{\Theta} \widehat{g}_n(S_n|\theta)\pi(\theta)\dt \theta}-\frac{{g}_n(S_n|\theta)\pi(\theta)}{\int_{\Theta} {g}_n(S_n|\theta)\pi(\theta)\dt \theta} \\
=&\left\{ \widehat{g}_n(S_n|\theta)-{g}_n(S_n|\theta)\right\}\frac{\pi(\theta)}{\int_{\Theta} {g}_n(S_n|\theta)\pi(\theta)\dt \theta} \frac{\int_{\Theta} {g}_n(S_n|\theta)\pi(\theta)\dt \theta}{\int_{\Theta} \widehat{g}_n(S_n|\theta)\pi(\theta)\dt \theta} \\
&-g_{n}(S_n|\theta) \pi(\theta)\left(\frac{1}{\int_{\Theta} g_{n}(S_n|\theta) \pi(\theta)\dt\theta}-\frac{1}{\int_{\Theta} \widehat{g}_{n}(S_n|\theta) \pi(\theta)\dt\theta}\right),
\end{flalign*} and apply the triangle inequality to obtain
\begin{flalign*}
\left|\widehat{\pi}_{}(\theta|S_n)-\pi_{}(\theta|S_n)\right| &\leq\left|\widehat{g}_n(S_n|\theta)-{g}_n(S_n|\theta)\right| \frac{\pi(\theta)}{\int \widehat{g}_n(S_n|\theta)\pi(\theta)\dt\theta}\\&+\frac{\left|\int_{\Theta} \widehat{g}_n(S_n|\theta)\pi(\theta)\dt\theta-\int_{\Theta} {g}_n(S_n|\theta)\pi(\theta)\dt\theta\right|}{\int_{\Theta} \widehat{g}_n(S_n|\theta)\pi(\theta)\dt\theta} \pi_{}(\theta|S_n).
\end{flalign*} Multiplying by $\|\theta\|^\gamma$, integrating both sides and applying
 equations~\eqref{eq:result1} and \eqref{eq:result2}, 
$$
\begin{aligned}
\int_{\Theta}\|\theta\|^\gamma\left|\widehat{\pi}_{}(\theta|S_n)-\pi_{}(\theta|S_n)\right| \dt \theta & \leq \frac{1}{\int_{\Theta} \widehat{g}_n(S_n|\theta)\pi(\theta)\dt\theta} \int\|\theta\|^\xi\left|\widehat{g}_n(S_n|\theta)-{g}_n(S_n|\theta)\right|\pi(\theta) \dt\theta\\&+\frac{\left|\int_{\Theta}\widehat{g}_n(S_n|\theta)\pi(\theta)\dt\theta-\int_{\Theta}{g}_n(S_n|\theta)\pi(\theta)\dt\theta\right|}{\int_{\Theta}\widehat{g}_n(S_n|\theta)\pi(\theta)\dt\theta}\int_{\Theta}\|\theta\|^\xi\pi(\theta|S_n)\dt\theta \\
& \leq O_p\left (1/m\right )+\frac{\int_{\Theta}{g}_n(S_n|\theta)\pi(\theta)\dt\theta}{\int_{\Theta}\widehat{g}_n(S_n|\theta)\pi(\theta)\dt\theta} O_p\left (1/m\right )\int_{\Theta}\|\theta\|^\xi\pi(\theta|S_n)\dt\theta \\
&=O_p\left(1/m\right).
\end{aligned}
$$
By equation \eqref{eq:new2}, $\int_{\Theta} \|\theta\|^\xi \pi(\theta|S_n)<\infty$, and 
 the first term in the second inequality is $O_p(1/m)$; 
 the second term is also $O_p(1/m)$ because 
   ${\int_{\Theta}{g}_n(S_n|\theta)\pi(\theta)\dt\theta}/
 {\int_{\Theta}\widehat{g}_n(S_n|\theta)\pi(\theta)\dt\theta}=1+O_p(1/m)$.
The stated result then follows.

\end{proof}

\begin{proof}[Proof of Corollary \ref{cor:one}]
The proof follows from Theorem~\ref{prop:bvm}. First, decompose $\bar{\theta}_n$ as
\begin{flalign*}
\bar\theta_n&=\int \theta \widehat{\pi}(\theta|S_n)\dt\theta=\int \theta\left\{ \widehat{\pi}(\theta|S_n)-\pi(\theta|S_n)\right\}\dt\theta+\int\theta\pi(\theta|S_n)\dt\theta ; 
\end{flalign*} 
 by the result of Theorem~\ref{prop:bvm1},
$$
\int \theta\left\{ \widehat{\pi}(\theta|S_n)-\pi(\theta|S_n)\right\}\le \int \|\theta\|| \widehat{\pi}(\theta|S_n)-\pi(\theta|S_n)|\dt\theta=O_p(1/[v_nm])
$$
so that
\begin{flalign*}
\overline\theta_n&=O_p(1/[v_nm])+\int\theta\pi(\theta|S_n)\dt\theta .
\end{flalign*}
Changing variables $\theta\mapsto t$ yields
	\begin{flalign*}
	\int_{\Theta} \theta\pi(\theta|S_n)\dt \theta=\int_{\mathcal{T}_n} \left(\theta_0+W_0^{-1}Z_n/v_n+W_0^{-1}t/{v_n}\right)\pi(t|S_n)\dt t; 
	\end{flalign*}
hence
	\begin{flalign*}
	W_0v_n(\bar{\theta}_n-\theta_0)-Z_n&=\int t \pi(t|S_n)\dt t+O_p(1/m)\\&= \int t \left[\pi(t|S_n)-N\{t;0,W_0\}\right]\dt t+\int t N\{t;0,W_0\}\dt t +O_p(1/m). \label{eq:post_mean1}
	\end{flalign*}
The second term on the right is zero. Therefore,	
	\begin{flalign*}
	\left|W_0v_n(\bar{\theta}_n-\theta_0)-Z_n\right|&=\left| \int t \left[\pi(t|S_n)-N\{t;0,W_0\}\right]\dt t\right| +O_p(1/m)\\&\leq \int \|t\|\left|\pi(t|S_n)-N\{t;0,W_0\}\right|\dt t+O_p(1/m)\\&=o_p(1)+O_p(1/m),
	\end{flalign*}
where the last line follows from Lemma~\ref{prop:bvm1}. 	
Recall the definition $Z_n=\nabla b(\theta_0)^{\intercal}\Delta^{-1}(\theta_0)\{b(\theta_0)-S_n\}$; under Assumption~\ref{ass:one},
	\begin{equation*}
	Z_n\Rightarrow N\left\{0, \nabla b(\theta_0)^{\intercal}\Delta^{-1}(\theta_0)V_0\Delta^{-1}(\theta_0)\nabla b(\theta_0)\right\}, 
	\end{equation*}
and the result follows.

\end{proof}

\begin{proof}[Proof of Theorem \ref{thm:acc}]
We first show that the result is satisfied if $\widehat{g}_n(S_n|\theta)$ in $\widetilde\alpha_n$ is replaced with the idealized counterpart ${g}_n(S_n|\theta)$,
 yielding the acceptance rate $$\alpha_n\asymp\frac{1}{v_n^{d_\theta}}\int q_n(\theta)g_n(S_n|\theta)\dt\theta.$$From the posterior concentration of $\pi(\theta|S_n)$ in Lemma \ref{prop:bvm1} and the restrictions on the proposal in Assumption \ref{ass:prop}, the acceptance probability $\alpha_n$ can be rewritten as
$$
\alpha_n\asymp \int_{}\mathbb{I}\left[\|\theta-\theta_0\|\leq\delta_n\right]q_n(\theta)g_n(S_n|\theta)/v_n^{d_\theta}\dt \theta+o_p(1), $$
for some $\delta_n=o(1)$ with $v_n\delta_n\rightarrow\infty$.

Following arguments mirroring those in the proof of Lemma \ref{prop:bvm1}, for any $\delta_n=o(1)$, on the set $\{\theta\in\Theta:\|t(\theta)\|\leq \delta_n v_n\}$, and disregarding $o(1)$ terms,
\begin{flalign*}
g_n\left(S_ { n }|\theta\right)/v_n^{d_\theta}\asymp \exp\left\{-t^{\intercal}(\theta)W_0^{-1}t(\theta)/2\right\},
\end{flalign*}
where $t(\theta):=W_0v_n\{\theta-\theta_0\}-Z_n$ (see the proof of Lemma \ref{prop:bvm1} for details). By construction, $t(\theta)$ is a one-to-one transformation of $\theta$ for fixed $\theta_0$ and $Z_n$. From the definition of the proposal, we can restrict $\theta$ to the set $$\{\theta\in\Theta:\|v_n(\theta-\theta_0)\|\leq\delta_nv_n\}\cap\{\theta\in\Theta:\theta=\mu_n+\sigma_n X\},$$with $\mathbb{E}[X]=0,\;\mathbb{E}[\|X\|^2]<\infty$. On this set, up to negligible terms,
\begin{equation}\label{eq:bsl_like}
g_n\left(S_ { n }|\theta\right)/v_n^{d_\theta}\asymp  \exp\left[-\left\{W_0v_n(\theta-\theta_0)-Z_n\right\}^{\intercal}W^{-1}_{0}\left\{W_0v_n(\theta-\theta_0)-Z_n\right\}/2\right].
\end{equation}
Define $r_n:=v_n\sigma_n$, $c_n^\mu:=\sigma^{-1}_n(\mu_n-S_n)$, and apply equation \eqref{eq:bsl_like} along with the change of variables $\theta\mapsto t=W_0v_n(\theta-\theta_0)$ to obtain
\begin{flalign*}
\alpha_n&\asymp\int_{}\mathbb{I}\left[\|v_n(\theta-\theta_0)\|\leq v_n\delta_n\right]q_n(\theta)g_n(S_n|\theta)/v_n^{d_\theta}\dt\theta\\&\asymp\int_{\|t\|\leq v_n\delta_n}r_n^{-1}{q\left(t/r_n-c_n^{\mu}\right) }{}\exp\left[-\left\{t-Z_n\right\}^{\intercal}M(\theta_0)\left\{t-Z_n\right\}/2\right]\dt t,
\end{flalign*}where the second equality makes use of the location-scale nature of the proposal. For $\delta_n v_n\rightarrow\infty$, $T_\delta:=\{t:\|t\|\leq\delta_n v_n\}\rightarrow\mathbb{R}^{d_\theta}$. Define $x(t):=t/r_{n}-c_{n}^{\mu}$ and the set $x(A):=\{x:x=x(t)\text{ for some }t\in A\}$. Then, by construction, $x(T_\delta)$ also converges to $\mathbb{R}^{d_\theta}$. Applying the change of variable $t\mapsto x$  yields
\begin{flalign}
\alpha_n&\asymp\int_{x(T_\delta)}r_n^{-1}q\left(x\right) \exp\left\{-r_n^2\left(x+c_{n}^{\mu}-Z_n/r_n\right)^{\intercal}M_{0}\left(x+c_{n}^{\mu}-Z_n/r_n\right)/2\right\}\dt x.\label{eq:acc1}
\end{flalign}

Applying part (i) of Assumption \ref{ass:prop} then  yields $$\underline{\alpha}_n\leq \alpha_n\leq \widehat{\alpha}_n,$$ where
\begin{flalign*}
\widehat{\alpha}_n&=\frac{C}{r_n^{}}\frac{|M(\theta_0)|^{1/2}}{(2\pi)^{d_\theta/2}}\int_{x(T_\delta)} \exp\left\{-r^2_n\left(x+c_{n}^{\mu}-Z_n/r_n\right)^{\intercal}M_{0}\left(x+c_{n}^{\mu}-Z_n/r_n\right)/2\right\}\dt x,\\
\underline{\alpha}_n&=\frac{\exp\left(-Z_n^\intercal{M}_0Z_n/2\right)}{{r_n^{}}}\frac{|M(\theta_0)|^{1/2}}{(2\pi)^{d_\theta/2}}\int_{x(T_\delta)}q(x) \exp\left\{-r_n^2\left(x+c_{n}^{\mu}\right)^{\intercal}M_{0}\left(x+c_{n}^{\mu}\right)/2\right\}\dt x.
\end{flalign*}
By part (ii) of Assumption \ref{ass:prop}, $r_n\rightarrow c_\sigma>0$;  by Assumption 1, $Z_n/r_n=O_p(1).$ Therefore, for $Z$ denoting a random variable whose distribution is the same as the limiting distribution of $Z_n$, by the dominated convergence theorem and  part (iii) of Assumption \ref{ass:prop}
\begin{flalign*}
\widehat{\alpha}_n&\rightarrow\frac{C}{r_0}\frac{|M(\theta_0)|^{1/2}}{(2\pi)^{d_\theta/2}}\int_{\mathbb{R}^d} \exp\left\{-r_0\left(x+c^{\mu}\right)^{\intercal}M_{0}\left(x+c_{}^{\mu}\right)/2\right\}\dt x,\\
\underline{\alpha}_n&\rightarrow\frac{\exp\left(-Z^\intercal M(\theta_0)Z/2\right)}{r_0}\frac{|M(\theta_0)|^{1/2}}{(2\pi)^{d_\theta/2}}\int_{\mathbb{R}^{d_\theta}}q(x) \exp\left\{-r_0^2\left(x+c^{\mu}\right)^{\intercal}M_{0}\left(x+c_{}^{\mu}\right)/2\right\}\dt x,
\end{flalign*}
in distribution as $n\rightarrow\infty$, where $c^\mu$ denotes a random variable whose distribution is the same as the limiting distribution of $\sigma^{-1}_n(\mu_n-S_n)$ and $r_0=\lim_n r_n$. 
By part (ii) of Assumption~\ref{ass:prop}, $c^\mu$ is finite except on sets of measure zero, ensuring that $\widehat{\alpha}_n=\Xi_p(1)$ and $\underline{\alpha}_n=\Xi_p(1).$ 
We have  $\alpha_n=\Xi_n(1)$
because  the above limits are $\Xi_p(1)$.

To deduce the stated result, we first bound $|\widetilde{\alpha}_n-\alpha_n|$ as
\begin{flalign*}
\left|\int q_n(\theta)\widehat{g}_n(S_n|\theta)\dt\theta-\int q_n(\theta){g}_n(S_n|\theta)\dt\theta\right|&\leq\int q_n(\theta)\left\{k(\theta)/m\right\}g_n(S_n|\theta)\dt\theta\\&={m}^{-1}\int\pi(\theta)g_n(S_n|\theta)\dt\theta\int\frac{q_n(\theta)}{\pi(\theta)}k(\theta)\pi(\theta|S_n)\dt\theta,
\end{flalign*}where the first inequality follows from  equation \eqref{eq:app3} in the proof of Theorem \ref{prop:bvm}, and the equality follows from reorganizing terms. Define $h_n(\theta):={q_n(\theta)}/{\pi(\theta)}$ and obtain
\begin{flalign*}
\int \frac{q_n(\theta)}{\pi(\theta)}k(\theta)\pi(\theta|S_n)\dt\theta=\int h_n(\theta)k(\theta)\pi(\theta|S_n)\dt\theta&\leq \left[\int h_n^2(\theta)\pi(\theta|S_n)\dt\theta\right]^{1/2}\left[\int k^2(\theta)\pi(\theta|S_n)\dt\theta\right]^{1/2}\\&\leq O_p(1)\left[\int \|\theta\|^{2\kappa}\pi(\theta|S_n)\dt\theta\right]^{1/2}\\&\leq O_p(1),
\end{flalign*}where the first inequality follows from Cauchy-Schwartz, the second from Assumption \ref{ass:prop} part (iv) and Assumption \ref{ass:propO} part (ii), while $\int \|\theta\|^{2\kappa}\pi(\theta|S_n)\dt\theta<\infty$ by hypothesis. Consequently,
\begin{flalign*}
|\widetilde\alpha_n-\alpha_n|=\left|\int q_n(\theta)\widehat{g}_n(S_n|\theta)\dt\theta-\int q_n(\theta){g}_n(S_n|\theta)\dt\theta\right|=O_p(1/m)
\end{flalign*}and the stated result follows from the behavior of $\alpha_n$ obtained in the first part of the result.
\end{proof}

\subsection{Lemmas\label{SS: lemmas}}
This section contains several lemmas used to prove the main results. The first lemma draws on elements from \cite{lehmann2006theory} and \cite{chernozhukov+h03} to demonstrate that the exact BSL posterior is asymptotically normal. We note that the simulated nature of the BSL likelihood implies that the above results are not directly applicable in our context.

\begin{lemma}\label{prop:bvm1}
	Recall $t:=W_0v_n(\theta-\theta_0)-Z_n;$ if Assumptions \ref{ass:one}-\ref{ass:two} are satisfied, and if Assumption \ref{ass:three} is satisfied with $p\ge\gamma\ge0$, then
	$$
	\int\|t\|^\gamma\left|\pi(t|S_n)-N\{t;0,W_0\}\right|\dt t=o_{p}(1).
	$$
\end{lemma}

\begin{proof}[Proof of Lemma \ref{prop:bvm1}]Recall the following definitions used in the proof of Theorem~\ref{prop:bvm}: $M_n(\theta):=\left[v_n^2{\Delta}^{}_n(\theta)\right]^{-1}$,  $M(\theta):=\Delta(\theta)^{-1}$, and $Q_n(\theta):=-v_n^2\{b(\theta)-S_n\}^\intercal M_n(\theta)\{b(\theta)-S_n\}/2.$ For an appropriately defined remainder term $R_n(\theta)$, consider the identity
\begin{flalign}
Q_n(\theta)-Q_n(\theta_{0})&=v_n^2\left\{b(\theta_0)-S_n\right\}^{\intercal}M_n(\theta_0)\nabla b(\theta_0)^{\intercal}(\theta-\theta_0)\nonumber\\&-\frac{v_n^2}{2}(\theta-\theta_0)\nabla b(\theta_0)^{\intercal}M_n(\theta_0)\nabla b(\theta_0)(\theta-\theta_0)+R_n(\theta)\nonumber
\\&=-\frac{1}{2}t^{\intercal}W_0^{-1}t+\frac{1}{2}Z_n^{\intercal}W_0^{-1}Z_n+R_n(\theta).\label{eq:new3}
\end{flalign}
To simplify notation, let $T_n:=\theta_0+W_0^{-1}Z_n/v_n$, $t_w:=W_0^{-1}t$, and define
\begin{flalign*}
\omega(t)&:= Q_n\left(T_n+t_w/v_n\right)-Q_n(\theta_0)-\frac{1}{2}Z_n^{\intercal}W_0^{-1}Z_n.
\end{flalign*}	
Applying \eqref{eq:new3}, we see that 
$$
\omega(t)=-\frac{1}{2}t^{\intercal}W_0^{-1}t+R_n(T_n+t_w/v_n).$$
Then, for $\mathcal{T}_n:=\{W_0v_n(\theta-\theta_0)-Z_n:\theta\in\Theta\}$,
\begin{flalign*}
\pi(t|S_n)&=\frac{\left|M_n\left(T_n+t_w/v_n\right)\right|^{1/2}\exp\left\{Q_n\left(T_n+t_w/v_n\right)\right\}\pi\left(T_n+t_w/v_n\right)}{\int_{\mathcal{T}_n} \left|M_n\left(T_n+t_w/v_n\right)\right|^{1/2}\left\{Q_n\left(T_n+t_w/v_n\right)\right\}\pi\left(T_n+t_w/v_n\right)\dt t}\\&=\frac{\left|M_n\left(T_n+t_w/v_n\right)\right|^{1/2}\exp\left\{Q_n\left(T_n+t_w/v_n\right)-Q_n(\theta_0)-\frac{1}{2}Z_n^{\intercal}W_0^{-1}Z_n\right\}\pi\left(T_n+t_w/v_n\right)}{\int_{\mathcal{T}_n} \left|M_n\left(T_n+t_w/v_n\right)\right|^{1/2}\left\{Q_n\left(T_n+t_w/v_n\right)-Q_n(\theta_0)-\frac{1}{2}Z_n^{\intercal}W_0^{-1}Z_n\right\}\pi\left(T_n+t_w/v_n\right)\dt t}\\&={\left|M_n\left(T_n+t_w/v_n\right)\right|^{1/2}\exp\left\{ w(t) \right\}\pi\left(T_n+t_w/v_n\right)}/{C_n},
\end{flalign*}where
$$
C_n=\int_{\mathcal{T}_n} |M_n\left(T_n+t_w/v_n\right)|^{1/2}\exp\left\{ \omega(t) \right\}\pi\left(T_n+t_w/v_n\right)\dt t.
$$Throughout the rest of the proof, unless otherwise specified, integrals are calculated over ${\mathcal{T}_n}$.

The stated result follows if
\begin{flalign*}
\int \|t\|^{\gamma}\left|\pi(t|S_n)-N\{t;0,W_0\}\right|\dt t&=C_n^{-1}J_n=o_p(1), 
\end{flalign*}
where
\begin{flalign*}
J_{n}&=\int\|t\|^{\gamma}\bigg{|}\left|M_n\left(T_n+\frac{t_w}{v_n}\right)\right|^{1/2}\exp\left\{\omega(t)\right\} \pi_{}\left(T_n+\frac{t_w}{v_n}\right)-\left|M(\theta_0)\right|^{1/2 } \exp \left\{-\frac{1}{2} t^{\intercal}W^{-1}_0 t\right\}C_{n}^{}\bigg{|} \dt t.
\end{flalign*}
However, $$J_{n}\leq J_{1n}+J_{2n},$$ where
\begin{flalign*}
J_{1n}&:= \int\|t\|^{\gamma}\left|\left|M_n\left(T_n+\frac{t_w}{v_n}\right)\right|^{\frac{1}{2}}\exp \left\{{\omega_{}(t)}{}\right\} \pi_{}\left(T_n+\frac{t_w}{v_n}\right)-|M(\theta_0)|^{\frac{1}{2}}\exp \left\{-\frac{1}{2} t^{\intercal} W_0^{-1}t\right\} \pi_{}\left(\theta_0\right)\right| \dt t\\
J_{2n}&:=\left|C_{n}-\pi(\theta_0)\right|\int \|t\|^{\gamma}|M(\theta_0)|^{1/2}\exp \left\{-\frac{1}{2} t^{\intercal} W_0^{-1} t\right\} \dt t .
\end{flalign*}
Therefore, if $J_{1n}=o_{p}(1)$ the result follows since, taking $\gamma=0$, $J_{1n}=o_{p}(1)$ implies that
\begin{flalign*}
\left|C_n-\pi(\theta_0) \right|&=\bigg{|}\int \left|M_n\left(T_n+\frac{t_w}{v_n}\right)\right|^{1/2}\exp\left\{\omega(t)\right\} \pi_{}\left(T_n+\frac{t_w}{v_n}\right)
\dt t\\&-\pi(\theta_0)\int |M(\theta_0)|^{1/2}\exp \left\{-\frac{1}{2} t^{\intercal} W_0^{-1} t\right\} \dt t\bigg{|}\\&=o_{p}(1),
\end{flalign*}which implies that $J_{2n}=o_{p}(1)$.

To demonstrate that $J_{1n}=o_p(1)$, we split $\mathcal{T}_n$ into three regions. For some $0<h<\infty$ and $\delta>0$, with $\delta=o(1)$: region 1: $ \|t\|\leq h$; region 2: $  h<\|t\|\leq \delta v_n$;  region 3: $  \|t\|\geq \delta v_n$.

\medskip

\noindent\textbf{\textbf{Region 1}:} Over this region the result follows if
$$
\|t\|^\gamma \left|\left|M_n\left(T_n+\frac{t_w}{v_n}\right)\right|^{1/2}\exp\left\{\omega(t)\right\} \pi_{}\left(T_n+\frac{t_w}{v_n}\right)-\pi(\theta_0) |M(\theta_0)|^{1/2}\exp \left\{-\frac{1}{2} t^{\intercal} W_0^{-1} t\right\}\right|=o_p(1).
$$
Note that,
\begin{flalign*}
\quad \sup_{\|t\|\leq h} \left\|M_n\left(T_n+\frac{t_w}{v_n}\right)-M(\theta_0)\right\|=o_{p}(1),\text{ and }\sup_{\|t\|\leq h}&\left|\pi\left(T_n+\frac{t_w}{v_n}\right)-\pi(\theta_0)\right|=o_{p}(1),
\end{flalign*}
where the first equation follows from Assumptions \ref{ass:two} and  \ref{ass:three}, and because $$T_n=\theta_0+W_0^{-1}Z_n/v_n=\theta_0+o_{p}(1),$$ since $Z_n=O_p(1)$ by Assumption \ref{ass:one}. Likewise, by Assumption \ref{ass:one},
$$
\sup_{\|t\|\le h}\left\|T_n+t_w/v_n-\theta_0\right\|=O_p(1/v_n)
$$
so that by the first part of Lemma \ref{lem:remain},
$$
\sup_{\|t\|\le h}|R_n(T_n+t_w/v_n)|=o_p(1).
$$
Hence, $J_{1n}=o_{p}(1)$ from these equivalences and the dominated convergence theorem.

\bigskip

\noindent\textbf{\textbf{Region 2}:}
For $\delta=o(1)$ and small enough, by Assumption 3,  $\sup_{h\le\|t\|\le \delta v_n}\|M_n\left(T_n+t_w/v_n\right)-M(\theta_0)\|^{}=o_{p}(1)$.
For $h$ large enough and $\delta=o(1)$,  we have the bound  ${J}_{1n}\leq C_{1n}+C_{2n}+C_{3n}$, where
\begin{flalign*}
C_{1n}:=&C\int_{h\leq \|t\| \leq \delta v_{n}}\|t\|^{\gamma}\exp(-t^{\intercal}W_0^{-1}t/2)\sup _{\|t\| \leq h}\left|\exp\left\{|R_n(T_n+t_w/v_n)|\right\}\left\{ \pi_{}\left(T_n+{t_w}/{v_n}\right)-\pi_{}\left(\theta_0\right)\right\}\right|\dt t \\
C_{2n}:=&C\int_{h\leq \|t\| \leq \delta v_{n}}\|t\|^{\gamma} \exp(-t^{\intercal}W_0^{-1}t/2)\exp\left\{|R_n(T_n+t_w/v_n)|\right\} \pi_{}\left(T_n+{t_w}/{v_n}\right) \dt t \\C_{3n}:=&C\pi_{}\left(\theta_0\right) \int_{h\leq \|t\| \leq \delta v_{n}}\|t\|^{\gamma}\exp(-t^{\intercal}W_0^{-1}t/2)\dt t .
\end{flalign*}

The first term  $C_{1n}=o_{p}(1)$ for any fixed $h$, so that $C_{1n}=o_{p}(1)$ for $h\rightarrow\infty$, by the dominated convergence theorem. For  $C_{3n}$, we have that for any $0\le\gamma\le2$ there exists some $h'$ large enough such that for all $h>h'$, and $\|t\|\ge h$ $$\|t\|^{\gamma}\exp\left(-t^{\intercal}M_{0}t\right)=O(1/h).$$ Hence, $C_{3n}$ can be made arbitrarily small by taking $h$ large and $\delta$ small enough.

The result follows if $C_{2n}=o_p(1)$. We show that, for some $C>0$, and all $h\le\|t\|\le\delta v_n$,  with probability converging to one (wpc1),
\begin{equation}\label{eq:bound1}
\exp(-t^{\intercal}W_0^{-1}t/2)\exp\left\{|R_n(T_n+t_w/v_n)|\right\}\pi(T_n+t_w/v_n)\le C\exp\left\{-t^{\intercal}W_0^{-1}t/4\right\}.
\end{equation}
If equation \eqref{eq:bound1} is satisfied, then $C_{2n}$ is bounded above by
\begin{flalign*}
C_{2n}\le &C\int_{h\le\|t\|\le\delta v_n}\|t\|^{\gamma}\exp\left\{-t^{\intercal}W_0^{-1}t/4\right\}\dt t,
\end{flalign*}which, again can be made arbitrarily small for some $h$ large and $\delta$ small. To demonstrate equation \eqref{eq:bound1}, first note that by continuity of $\pi(\theta)$, Assumption \ref{ass:three}, $\pi(T_n+t_w/v_n)$ is bounded over $\{t:h\le \|t\|\le\delta v_n\}$ so that it may be dropped from the analysis. Now, since $\|T_n-\theta_0\|=o_p(1)$, for any $\delta>0$, $\|T_n+t_w/v_n-\theta_0\|<2\delta$ for all $\|t_w\|\le\delta v_n$ and $n$ large enough. Therefore, by Lemma \ref{lem:remain}, there exists some $\delta>0$ and $h$ large enough so that (wpc1)
$$
\sup_{h\le\|t\|\le\delta v_n}|R_n(T_n+t_w/v_n)|\le \frac{1}{4}\|t-Z_n\|^2\lambda_{\text{min}}\left\{W_0\right\}.
$$Since $Z_n=O_p(1)$, we have $Z_n^{\intercal}W_0^{-1}Z_n\le C\|Z_n\|^2=O_p(1)$, so that, for some $C>0$, wpc1,
$$
\exp\{\omega(t)\}\leq \exp\{-t^{\intercal}W_{0}^{-1}t+|R_n(T_n+t_w/v_n)|\}\leq C\exp\left(-t^{\intercal}W_0^{-1}t/4\right),
$$and the result follows.

\medskip

\noindent\textbf{\textbf{Region 3}:} For $\delta v_n$ large,
$$
\int_{\|t\|\ge \delta v_n}\|t\|^{\gamma}N\{t;0,W_0\}\dt t$$ can be made arbitrarily small and is therefore dropped from the analysis. Consider
\begin{align*}
\tilde{J}_{1n}& :=\int_{\|t\|\ge \delta v_n}\|t\|^{\gamma} \left| M_{n}\left(t / v_{n}+S_{n}\right)\right|^{1/2} \exp\{\omega(t) \} \pi\left(t / v_{n}+S_{n}\right)\dt t , \\
& = v_{n}^{d_\theta+\gamma}\int_{\|\theta-T_n\|\ge \delta }\|\theta-T_n\|^{\gamma}|M_n(\theta)|^{1/2}\exp\left\{Q_n(\theta)\right\} \pi\left(\theta\right)\dt \theta,
\end{align*}
by using the change of variables $\theta=T_n+t_w/v_n$. Now,
\begin{flalign*}
\tilde{J}_{1n}& =\exp\left\{Q_n(\theta_0)\right\}v_{n}^{d_\theta+\gamma}\int_{\|\theta-T_n\|\ge \delta }\|\theta-T_n\|^{\gamma}|M_n(\theta)|^{1/2}\exp\left\{Q_n(\theta)-Q_n(\theta_0)\right\} \pi\left(\theta\right)\dt \theta,
\end{flalign*}
and note that $\exp\{Q_n(\theta_0)\}=O_{p}(1)$ because $Q_n(\theta_0)=O_{p}(1)$ by Assumptions \ref{ass:one} and \ref{ass:two}. 

Define $Q(\theta):=-\{b(\theta)-b(\theta_0)\}^{\intercal}M(\theta)\{b(\theta)-b(\theta_0)\}/2$ and note that $Q(\theta_0)=0$ by virtue of Assumption \ref{ass:four}(i) and positive-definiteness of $\Delta(\theta_0)$ (Assumption \ref{ass:two}(ii)). For any $\delta>0$,
\begin{flalign*}
\sup_{\|\theta-\theta_0\|\ge \delta}\frac{1}{v^2_n}\left\{Q_n(\theta)-Q_n(\theta_0)\right\}\leq& \sup_{\|\theta-\theta_0\|\ge \delta}2|v_n^{-2}Q_n(\theta)-Q(\theta)|+\sup_{\|\theta-\theta_0\|\ge \delta}\left\{Q(\theta)-Q(\theta_0)\right\}.
\end{flalign*}
From Assumptions~\ref{ass:one} and \ref{ass:two}, the first term converges to zero in probability. From Assumption~\ref{ass:two}(iii), for any $\delta>0$ there exists an $\epsilon>0$ such that
$$
\sup_{\|\theta-\theta_0\|\ge \delta}\left\{Q(\theta)-Q(\theta_0)\right\}\le -\epsilon.
$$
Hence, 
\begin{equation}
\label{eq:expconv}
\lim_{n\rightarrow\infty}P^{(n)}_0\left[\sup_{\|\theta-\theta_0\|\geq \delta}\exp\left\{Q_n(\theta)-Q_n(\theta_0)\right\}\leq \exp(-\epsilon v_n^2)\right]=1.
\end{equation}
Use $T_n=\theta_0+O_p(1/v_n)$, the definition $M_n(\theta)=v_n^2\Delta_n(\theta)$, and equation \eqref{eq:expconv} to  obtain
\begin{align*}
\tilde{J}_{1n} & = \{1+o_p(1)\}\exp\left\{Q_n(\theta_0)\right\}v_{n}^{d_\theta+\gamma}\int_{\|\theta-\theta_{0}\|\ge \delta }|v^2_n\Delta_n(\theta)|^{-1/2}\|\theta-\theta_0\|^{\gamma}\pi\left(\theta\right)\exp\{Q_n(\theta)-Q_n(\theta_0)\}\dt \theta\\&\leq O_p(1) \exp\left(-\epsilon v_n^2\right)v_{n}^{d_\theta+\gamma}\int_{\|\theta-\theta_0\|\ge \delta }|v^2_n\Delta_n(\theta)|^{-1/2}\|\theta-\theta_0\|^{\gamma}\pi\left(\theta\right)\dt \theta\\&\leq  O_p\left\{\exp\left(-\epsilon v_n^2\right)v_{n}^{d_\theta+\gamma}\right\}
\\&=o_p(1); 
\end{align*}where 
the third inequality follows from the moment hypothesis in Assumption~\ref{ass:three}.
\end{proof}

The following result is a consequence of Proposition 1 in \cite{chernozhukov+h03}.
\begin{lemma}\label{lem:remain} Under Assumptions \ref{ass:one}-3, and for $R_n(\theta)$ as defined in the proof of Lemma \ref{prop:bvm1}, for each $\epsilon>0$ there exists a sufficiently small $\delta>0$ and $h>0$ large enough, such that
$$
\limsup_{n\rightarrow\infty}\text{Pr}\left[\sup_{h/v_n\le \|\theta-\theta_0\|\le\delta}\frac{|R_n(\theta)|}{1+n\|\theta-\theta_0\|^2}>\epsilon\right]<\epsilon
$$	and
$$
\limsup_{n\rightarrow\infty}\text{Pr}\left[\sup_{ \|\theta-\theta_0\|\le h/v_n}{|R_n(\theta)|}>\epsilon\right]=0.
$$	
\end{lemma}
\begin{proof}
The result is a specific case of Proposition~1 in \cite{chernozhukov+h03}. Therefore, it is
only necessary to verify that their sufficient conditions are satisfied in our context.

Assumptions (i)-(iii) in their result follow
 directly from Assumptions \ref{ass:four} and \ref{ass:two}, and the normality of $v_n\{b(\theta_0)-S_n\}$ in 
 Assumption~\ref{ass:one}. Therefore, all that remains is to verify their Assumption (iv).

Defining ${g}_n(\theta)=b(\theta)-S_n$, their Assumption~(iv) is stated as follows: for any $\epsilon>0$, there is a $\delta>0$ such that
$$
\limsup_{n\rightarrow\infty}\text{Pr}\left\{\sup_{\|\theta-\theta'\|\le\delta}\frac{v_n\|
\{g_n(\theta)-g_n(\theta')\}-\{\mathbb{E}\left[g_n(\theta)\right]-\mathbb{E}
\left[g_n(\theta')\right]\}}{1+v_n\|\theta-\theta'\|}>\epsilon\right\}<\epsilon . 
$$
In our context, this condition is always satisfied:  for $g_n(\theta)=b(\theta)-S_n$, and all $n$
\begin{flalign*}
{\|\{g_n(\theta)-g_n(\theta')\}-\{\mathbb{E}\left[g_n(\theta)\right]-\mathbb{E}\left[g_n(\theta')\right]\}}\|&={\|\{b(\theta)-b(\theta')\}-\{[b(\theta)-b_0]-[b(\theta')-b_0]\}}\|\\&=0.
\end{flalign*}
\end{proof}

The following result is used in the proof of Theorem \ref{prop:bvm1} and is an intermediate result of Theorem 1 in \cite{hsu2012tail}.

\begin{lemma}[Theorem 1, \cite{hsu2012tail}]\label{lem:propO}
Suppose $x=(x_1,\dots,x_{d})^\intercal$ is a random vector such that for some $\mu\in\mathbb{R}^d$ and some $\sigma\ge 0$,
$$
\mathbb{E}\left[\exp \left\{\alpha^{\intercal}(x-\mu)\right\}\right] \leq \exp \left(\|{\alpha}\|^{2} \sigma^{2} / 2\right),
$$for all $\alpha\in\mathbb{R}^d$. For $M\in\mathbb{R}^{d\times d}$ a positive-definite and symmetric matrix such that $M:=A^{\intercal} A$, for  $0\leq \eta<1/(2\sigma^2\|M\|)$,
$$
\begin{aligned}
\mathbb{E}\left[\exp \left(\eta\|A x\|^{2}\right)\right] &
& \leq \exp \left\{\sigma^{2} \operatorname{tr}(M) \eta+\frac{\sigma^{4} \operatorname{tr}\left(M^{2}\right) \eta^{2}+\|A \mu\|^{2} \eta}{1-2 \sigma^{2}\|M\|\eta}\right\}
\end{aligned}.
$$

\end{lemma}

\end{document}